\documentclass[12pt, reqno]{amsart}

\usepackage{amssymb, amscd, amsthm, mathtools}
\usepackage{url,amssymb,enumerate,colonequals}
\usepackage[margin = 1in]{geometry}

\usepackage{pdfsync}
\usepackage{multirow}
\usepackage[usenames,dvipsnames]{xcolor}
\usepackage{url}
\usepackage{microtype}
\usepackage{tikz-cd}
\usepackage[all]{xy}
\usetikzlibrary{matrix}
\usepackage{float}
\usepackage{algorithm}
\usepackage{todonotes}
\usepackage{algpseudocode}
\usepackage{amsmath}

\floatstyle{boxed}
\restylefloat{figure}

\def\ZZ{\mathbb{Z}}
\def\PP{\mathbb{P}}

\newcommand{\F}{\mathbb{F}}

\newcommand{\calO}{\mathcal{O}}
\newcommand{\calP}{\mathcal{P}}

\newcommand{\fC}{\mathfrak{C}}
\newcommand{\cP}{\mathcal{P}}
\newcommand{\supp}{{\rm supp}}

\def \H {{\mathcal H}}

\newcommand{\cc}{\mathfrak{c}}

\newcommand{\cL}{\mathcal{L}}

\def\cQ{{\mathcal Q}}
\def\fC{{\mathfrak{C}}}
\newcommand{\bj}{{\bf j}}
\newcommand{\bt}{{\bf t}}
\newcommand{\bx}{{\bf x}}

\newcommand{\Ga}{\alpha}
\newcommand{\Gb}{\beta}
\newcommand{\Gl}{\lambda}
\newcommand{\Gk}{\kappa}
\newcommand{\Pin}{{P_\infty}}

\newcommand{\cB}{\mathcal{B}}
\newcommand{\GL}{\mathrm{GL}}
\newcommand{\PGL}{\mathrm{PGL}}
\newcommand{\AGL}{\mathrm{AGL}}

\DeclareMathOperator{\Div}{\mathrm{Div}}
\DeclareMathOperator{\Gal}{\mathrm{Gal}}

\DeclareMathOperator{\Tr}{\mathrm{Tr}}

\DeclareMathOperator{\Ker}{\mathrm{Ker}}
\DeclareMathOperator{\Aut}{\mathrm{Aut}}

\DeclareMathOperator{\ddiv}{\mathrm{div}}
\DeclareMathOperator{\ev}{\mathrm{ev}}

\DeclareMathOperator{\by}{\textbf{y}}

\DeclareMathOperator{\bm}{\textbf{m}}

\DeclareMathOperator{\bfi}{\textbf{i}}
\def\cX{{\mathcal X}}

\newtheorem{theorem}{Theorem}[section]
\newtheorem{lemma}[theorem]{Lemma}
\newtheorem{corollary}[theorem]{Corollary}
\newtheorem{proposition}[theorem]{Proposition}

\newtheorem{Maintheorem}[theorem]{\bf Main Theorem}

\theoremstyle{definition}
\newtheorem{definition}[theorem]{Definition}

\newtheorem{example}[theorem]{Example}

\newtheorem{remark}[theorem]{Remark}

\numberwithin{equation}{subsection}

\begin{document}
\title[Fast encoding of AG codes]{Encoding of algebraic geometry codes with quasi-linear complexity $O(N\log N)$}

\author{Songsong Li}
\address{School of Electronic Information and Electrical Engineering, Shanghai Jiao Tong University, Shanghai,  China}
\email{songsli@sjtu.edu.cn}

\author{Shu Liu}
\address{National Key Laboratory on Wireless Communications, University of Electronic Science and Technology of China, Chengdu, China}
\email{shuliu@uestc.edu.cn}

\author{Liming Ma}
\address{School of Mathematical Sciences, University of Science and Technology of China, Hefei, China}
\email{ lmma20@ustc.edu.cn}

\author{Yunqi Wan}
\address{Huawei,  China}
\email{wanyunqi@huawei.com}

\author{Chaoping Xing}
\address{School of Electronic Information and Electrical Engineering, Shanghai Jiao Tong University, Shanghai,  China}
\email{ xingcp@sjtu.edu.cn}

\subjclass[]{}
		
\maketitle	
\begin{abstract}
Fast encoding and decoding of codes have been always an important topic in code theory as well as complexity theory. Although encoding is easier than decoding in general, designing an encoding algorithm of codes of length $N$ with quasi-linear complexity $O(N\log N)$ is not an easy task. Despite of the fact that algebraic geometry codes were discovered in the early of 1980s, encoding algorithms of algebraic geometry codes with quasi-linear complexity $O(N\log N)$  have not been found except for the simplest algebraic geometry codes--Reed-Solomon codes. 
 The best-known encoding algorithm of algebraic geometry codes based on a class of plane curves has quasi-linear complexity at least $O(N\log^2 N)$ ( {IEEE Trans. Inf. Theory 2021}). In this paper, we design an encoding algorithm of algebraic geometry codes with quasi-linear complexity  $O(N\log N)$. Our algorithm works well for a large class of algebraic geometry codes based on both plane and non-plane curves.

The main idea of this paper is to generalize the divide-and-conquer method from the fast Fourier Transform over finite fields to algebraic curves. More precisely speaking, suppose we consider encoding of algebraic geometry codes based on an algebraic curve ${\mathcal X}$ over $\mathbb{F}_q$. We first consider a tower of Galois coverings ${\mathcal X}={\mathcal X}_0\rightarrow{\mathcal X}_1\rightarrow\cdots\rightarrow{\mathcal X}_r$ over a finite field $\mathbb{F}_q$, i.e., their function field tower $\mathbb{F}_q({\mathcal X}_0)\supsetneq\mathbb{F}_q({\mathcal X}_{1})\supsetneq\cdots \supsetneq\mathbb{F}_q({\mathcal X}_r)$ satisfies that each of extension $\mathbb{F}_q({\mathcal X}_{i-1})/\mathbb{F}_q({\mathcal X}_i)$ is a Galois extension and the extension degree $[\mathbb{F}_q({\mathcal X}_{i-1}):\mathbb{F}_q({\mathcal X}_i)]$  {is a constant}. Then encoding of an algebraic geometry code based on ${\mathcal X}$ is reduced to the encoding of an algebraic geometry code based on ${\mathcal X}_r$. As a result, if there is an encoding algorithm of the algebraic geometry code based on ${\mathcal X}_r$ with quasi-linear complexity, then encoding of the algebraic geometry code based on ${\mathcal X}$ also has quasi-linear complexity.
\end{abstract}

\section{Introduction}\label{sec: introduction}	
Fast encoding and decoding of codes with good parameters is a major topic in code theory as well as complexity theory. We focus on fast encoding in this paper. Although encoding is easier than decoding in general, designing an encoding algorithm of codes of length $N$ with low complexity, particularly with quasi-linear complexity $O(N\log N)$ is not easy.

In the topic of encoding of algebraic geometry codes, people mainly studied encoding of algebraic geometry codes based on plane curves. Some of these codes based on plane curves include Reed-Solomon (RS for short) codes, elliptic codes, Hermitian codes, etc.  Due to the excellent parameters of algebraic geometry codes, it is urgent to design fast encoding and decoding algorithms to meet practical applications. The current paper moves one step forward on the fast encoding of algebraic geometry codes by designing algorithms of quasi-linear time  $O(N\log N)$.


Let $\F_q$ be a finite field of cardinality $q$. Let $\cX/\F_q$ be an algebraic curve defined over  $\F_q$. Given a Riemann-Roch space $\cL(D)$ for some divisor $D$ (see Subsection 2.1 for precise definition) and a set of rational points on $\cX$,  denoted by $\cP=\{P_1, P_2,\dots, P_N\}$ such that $\supp(D)\cap \cP=\emptyset$, the multipoint evaluation (MPE for short) of functions in $\cL(D)$ at the set $\cP$ is defined as
\[\ev_{\cP}: \cL(D)\rightarrow \F_q^N; f\mapsto \left(f(P_1),f(P_2),\dots,f(P_N)\right).\]
The algebraic geometry (AG for short) code $C(\calP, D)$ is defined to be $$C(\calP, D)=\{\ev_{\cP}(f) \mid f\in\cL(D)\}.$$

In the case where  $\cX$ is the projective line and $D=(k-1)P_{\infty}$ is the $(k-1)$-multiples of a single point $P_{\infty}$ for some positive integer $k\leq N$, the algebraic geometry codes given by the above encoding are just Reed-Solomon codes (RS codes for short). In the case of RS codes, assume that the $N$ points correspond to $N$ elements $\alpha_1,\alpha_2,\dots,\alpha_N$ in $\F_q$, respectively. Then $\cL(D)=\F_q[x]_{<k}$ is the space of polynomials over $\F_q$ of degree less than $k$ and
\[\ev_{\calP}: \F_q[x]_{<k}\rightarrow \F_q^N; f(x)\mapsto \left(f(\alpha_1),f(\alpha_2),\dots,f(\alpha_N)\right).\]

The encoding process of algebraic geometry codes from the message space $\F_q^k$ into the codeword space  is decomposed into two steps:
\begin{itemize}
\item[(i)] Step 1: choose an $\F_q$-basis $\cB=\{f_1,f_2,\dots,f_k\}$ of the Riemann-Roch space $\cL(D)$. Then we map a message $\bm=(m_1,m_2,\dots,m_k)\in \F_q^k$ to a function $f_{\bm}=\sum_{i=1}^km_if_i\in \cL(D)$.
\item[(ii)] Step 2: evaluate $f_{\bm}$ at the multipoint set $\calP$ to get the codeword $\cc=\ev_{\calP}(f_{\bm})$.
\end{itemize}
As long as a basis $\cB$ of $\cL(D)$  is found, encoding of algebraic geometry codes mainly consists of the MPE in Step 2. Thus, we regard MPE algorithms as encoding algorithms for AG codes since the bases of the Riemann-Roch spaces of algebraic curves that we are interested in are all explicitly given.

 A naive encoding algorithm of algebraic geometry codes  costs $O(N^2)$ operations: by firstly pre-computing a generator matrix $G_{k\times N}\in \F_q^{k\times N}$ whose rows are $\ev_{\calP}(f_i)$ for all $1\le i\le k$, then encoding of a message $\bm$ can be done in $O(N^2)$ operations in $\F_q$ by directly computing $\cc=\bm G$. However, one could get faster encoding algorithms via fast multipoint evaluation (FMPE for short).
 
 Below, we briefly review some related work on the fast encoding of algebraic geometry codes.

\subsection{Related work}  The simplest AG codes are Reed-Solomon codes. Some other well-known algebraic geometry codes include those from plane curves such as elliptic and Hermitian curves. As stated above, encoding RS codes corresponds to the MPE of polynomials in $\F_q[x]_{<N}$. The authors of \cite{BM74} gave a generic quasi-linear $O(M(N)\log N)$ algorithm for univariate MPE via polynomials modular arithmetics, where $M(N)$ stands for the cost of multiplying two polynomials of degree less than $N$ (one can also refer to \cite[Corollary 10.8]{von13}). By taking the currently best result $M(N)=O(N\log N4^{\log ^*N})$ given in \cite{HvDH19}, the univariate MPE will cost $O(N\log^2N4^{\log ^*N})$ operations in the underlying field $\F_q$. However, for sets $\calP$ with good structures, then the univariate MPE can be done in time $O(N\log N)$ via the fast Fourier transform (FFT for short). For instance, if $\calP$ is a multiplicative subgroup of $\F_q^*$  \cite{cooley1965, Pollard71} or an additive subgroup \cite{lin2014novel} and $N=|\cP|$ is $O(1)$-smooth, i.e., all prime factors of $N$ are constant ( {we sometimes simply call an $O(1)$-smooth integer a smooth integer}), then the univariate MPE on $\cP$ costs $O(N\log N)$ operations. The fast Fourier transform on multiplicative or additive subgroups of $\F_q$ is called multiplicative and additive FFT, respectively. It is worth mentioning that the successful generalization of FFT to the additive case is due to representations of polynomials in $\F_q[x]_{<N}$ under a new basis \cite{lin2014novel}. A fast MPE on $\cP$ with complexity $O(N\log N)$ directly gives an encoding algorithm for RS codes with complexity $O(N\log N)$ \cite{jj76, lin2014novel, lin16fft}. Therefore, in the encoding AG codes, a well-chosen evaluation set and a specific basis of the function space may lead to faster algorithms.

Apart from RS codes, few quasi-linear time encoding algorithms exist in the literature for algebraic geometry codes from plane curves. To the best of our knowledge, none of these algorithms can run in time $O(N\log N)$ operations of $\F_q$. Recently, quasi-linear time encoding algorithms for some algebraic geometry codes were proposed in \cite{BRS20}. The authors considered algebraic geometry codes from $C_{ab}$ plane curves. A $C_{ab}$ curve is defined by an equation $H(x,y):=\sum_{i,j}a_{ij}x^iy^j=0$. For such a curve, there is a unique common pole of $x$ and $y$, denoted by $P_{\infty}$.  {The parameters $a, b$ are corresponding to $\nu_{P_{\infty}}(x)=-a$ and $\nu_{P_{\infty}}(x)=-b$, respectively.}  The Riemann-Roch space $\cL(\lambda P_{\infty})$ has a basis $\{x^iy^j \mid ia+jb\leq \lambda,j\leq a-1\}$. Thus any $f\in \cL(\lambda P_{\infty})$ can be represented as $f=\sum_{i<a} f_i(x)y^i$ with $f_i(x)\in\F_q[x]$. Then the MPE of $f$ can be computed by repeating the fast univariate MPE two times, i.e., one can first compute the MPEs of each $f_i(x)$ at the set $\calP_x$ consisting of all $x$-coordinates of multipoints, and then for each $\alpha\in \calP_x$, compute the MPE of $f(\Ga,y)=\sum_{i<a} f_i(\Ga)y^i$ at the subset of multipoints with $x$-coordinate $\alpha$. This idea is simple and was first proposed in \cite{YB92}, and also employed in \cite{MOS01}, but these two works did not give the explicit complexity analysis. The authors of \cite{BRS20} discussed the complexity in different cases. If a curve $H/\F_q$ has an evaluation set with maximal semi-grid properties, then they showed that the two-layer MPE can be done in quasi-linear time. As they take advantage of fast univariate MPE in \cite{von13}, the final encoding complexity is $O(M(N)\log N)>O(N\log^2 N)$.

The earliest work on efficient encoding of algebraic geometry codes based on non-plane curves was given in \cite{TVG13,Lop00} where the authors showed that algebraic geometry codes based on modular curves have a polynomial time encoding. In \cite{SAKSD01}, the authors presented an algorithm with complexity $O(N^3)$ by obtaining a generator matrix of algebraic geometry codes on function fields of the Garcia-Stichtenoth tower. 
The first efficient encoding algorithm for algebraic geometry codes exceeding the GV bound with complexity less than $O(N^2)$ was presented by Narayanan and Weidner \cite{NW19}  where an encoding algorithm for algebraic geometry codes exceeding the GV bound with complexity $O(N^{w/2})$ (here $w$ is the matrix multiplicative exponent)  was designed. 



Many recent works considered the FMPE of multivariate polynomials \cite{BGKM22, BGGKU22, van20}. Although encoding of algebraic geometry codes also involves multipoint evaluation of multivariate polynomials, their variables are not independent and must satisfy certain algebraic relations. Thus, MPE on algebraic curves is a different topic from MPE of multivariate polynomials. Thus, we will not follow the work on fast MPE of multivariate polynomials but try to get faster algorithms for MPE of functions of algebraic curves.

\subsection{Our results and comparisons}
As we have already seen, so far there are no encoding algorithms in literature with quasi-linear time complexity $O(N\log N)$ for algebraic geometry codes (except for RS codes) for both plane and non-plane curves. The main contribution of this paper is to present encoding algorithms for algebraic geometry codes with quasi-linear time complexity $O(N\log N)$  by digging out the algebraic properties of algebraic curves with Galois covering.

The present paper contains two major results, one is for encoding of algebraic geometry codes from plane curves. The second one is for encoding algebraic geometry codes from non-plane curves. 
\begin{Maintheorem}\label{main:1}
Encoding of $q$-ary algebraic geometry codes of length $N$  based on the plane curve $A(y)=u(x)$ runs in time $O(N\log N)$ if $q$ or $q-1$ is $O(1)$-smooth, where $(A(y),u(x))$ is a pair of polynomials in $\F_q[y]$ and $\F_q[x]$ with $\gcd(\deg(A),\deg(u))=1$ and satisfies either (i) $A(y)=y^m$ with $m\mid(q-1)$ and $u(x)$ is square-free; or (ii) $A(y)$ is a $p$-linearized polynomial with $p$ being the characteristic of $\F_q$.
\end{Maintheorem}
Please refer to Corollaries \ref{cor:k} and \ref{cor:at} for the above Main Theorem 1.

\begin{remark}\label{rmk:normtrace curve}
 Let $q=\kappa^r$ and let $e$ be a divisor of $(\kappa^r-1)/(\kappa-1)$. Then the plane curve defined by $y^{\kappa^{r-1}}+y^{\kappa^{r-2}}+\cdots+y^\kappa+y=x^e$ is called a norm-trace curve over $\F_q$. If $r=2$, then this norm-trace curve is the well-known Hermitian curve given by $y^\kappa+y=x^{\kappa+1}$.
It was shown in \cite{BRS20} that the major candidates for $C_{ab}$ curve satisfying the $O(M(N)\log N)$ encodable requirements are norm-trace and other Hermitian-like curves. It is clearly seen that our plane curves defined in Main Theorem \ref{main:1} contain the class of norm-trace and other Hermitian-like curves. In other words, Main Theorem \ref{main:1} shows that algebraic geometry codes from norm-trace and other Hermitian-like curves allow encoding with complexity $O(N\log N)$, while it was shown in  \cite{BRS20} that algebraic geometry codes from norm-trace and other Hermitian-like curves have encoding complexity $O(M(N)\log N)>O(N\log^2 N)$.
\end{remark}

\begin{remark} The class of plane curves defined in Main Theorem \ref{main:1} contains a large family of well-known curves such as Kummer curves and Artin-Schreier curves. For instance, this class contains the Hermitian curves and their coverings. By a covering of the Hermitian curve, we mean a curve defined by $\prod_{\Ga\in W}(y-\Ga)=x^e$, where $W$ is an additive subgroup of the group $\{\Gb\in\F_q:\; \Gb^\kappa+\Gb=0\}$ and $e$ is a divisor of $\kappa+1$. Note that a covering of the Hermitian curve is also maximal, i.e., it achieves the Hasse-Weil bound (see Subsection 2.1 for definition).
Thus, algebraic geometry codes from coverings of the Hermitian curve also allow encoding with complexity $O(N\log N)$. Some other curves contained in  the class of curves defined in Main Theorem \ref{main:1} include elliptic curves, hyperelliptic curves, and norm-trace curves (see Example \ref{ex:5.2} for the detail)
\end{remark}

\begin{Maintheorem}\label{main:2} If $q$ or $q-1$ is $O(1)$-smooth, then there exists a family of non-plane algebraic geometry codes over $\F_q$ of length $N=q^{(n+1)/2}\rightarrow\infty$ and $\sqrt{q}\ge 2n$ with encoding complexity $O(N\log N)$.
\end{Maintheorem}

Please refer to Section 6  for the above Main Theorem 2.

\begin{remark} 
It was shown in \cite{NW19} that there exists a deterministic algorithm to encode a message into an algebraic geometry code based on non-plane curves in $O(N^{w/2})$ time, where $w$ is the matrix multiplicative exponent with $w<2.373$, while our algorithm can encode a message with quasi-linear complexity $O(N\log N)$  for algebraic geometry codes from non-plane curves.
\end{remark}

\subsection{Our techniques}\label{subsec:tech}
Algebraic geometry codes are based on algebraic curves over finite fields with many rational points. However, it is often more convenient to use the language of function fields rather than algebraic curves. It is a well-known fact in the community of arithmetic number theory that algebraic curves over finite fields are equivalent to function fields of one variable over finite fields, namely, there is a one-to-one correspondence between algebraic curves over finite fields and function fields of one variable over finite fields. In this paper, we choose to use the language of function fields with occasional use of the language of algebraic curves.

Due to the fast univariate multipoint evaluation via the FFT, an RS code of length $s$ is $O(s\log s)$ encodable. In general, we call a function field $F/{\F_q}$ FMPE-friendly if there is an MPE problem of length $s$ on $F$ that can be solved in $O(s\log s)$ operations of $\F_q$.  To make the definition meaningful, we assume $s=\Theta(|N(F)|)$, where $|N(F)|$ is the total number of rational places of $F$ (refer to Subsection~\ref{sec:ff}). If $s$ is too smaller than $|N(F)|$, it is meaningless to concern the AG codes of length $s$ from $F$. It is known that $\F_q(x)$ is FMPE-friendly if $q$ or $q-1$ is $O(1)$-smooth \cite{Pollard71,lin2014novel}. 

Our idea is to construct algebraic extensions $E$ over an FMPE-friendly function field $F$ such that the extended MPE problem of length $N=[E: F]\cdot s$ on $E$ can be efficiently reduced to the MPE problem on $F$. Then the extended MPE on $E$ can be solved in $O(N\log N)$. The reduction is motivated by the divide-and-conquer method in FFT. To achieve this, we make some assumptions on the extension $E/F$ and then show in Section~\ref{sec:inst} that these assumptions are satisfied for a large class of function fields containing the most well-known curves. 

To be more precise, let $F$ be an FMPE-friendly function field such that an MPE on $F$ of length $s$ can be done in $O(s\log s)$ operations. Assume the multipoint set and the Riemann-Roch space of the MPE are $\cQ=\{Q_1,\ldots, Q_s\}$ and $\cL(n Q_{\infty})$, respectively, where $Q_{\infty}, Q_1,\ldots, Q_s$ are $s+1$ distinct rational places of $F$. Suppose that $E/F$ is a finite extension satisfying 
\begin{enumerate}
\item The place $Q_{\infty}$ is totally ramified in $E/F$.
\item The $s$ places in $\cQ$ are splitting completely in $E/F$.
\end{enumerate}
Let $P_{\infty}$ be the unique place of $E$ lying over $Q_{\infty}$ and $\cP$ be the set of places of $E$ lying over places in $\cQ$. Then $|\cP|=N=ms$, where $m=[E: F]$. Let $\lambda$ be an integer not greater than $nm$. The extended MPE on $E$ is the linear map $\ev_{\cP}$ from $\cL(\lambda P_{\infty})$ to $\F_q^N$. In the following, we reduce the MPE $\ev_{\calP}(f)$ of an $f\in\cL(\lambda P_{\infty})$ to the MPE $\ev_{\cQ}$ of $m$ functions in $\cL(nQ_{\infty})$ in time $O(N\log m)$ if the field extension $E/F$ satisfies certain properties.

Suppose $E/F$ is an abelian extension and the extension degree $m=\prod_{i=1}^r p_i$ is an $O(1)$-smooth number, where all $p_i$ are constant primes. Then there is a chain of subgroups of $G=\Gal(E/F)$:
\[G_0=\{1\},\  G_{i-1}\leq G_{i}\ \text{of\ index}\ [G_i : G_{i-1}]=p_i,\  \text{for}\ i=1,\ldots,r.\]
Thus each $G_i$ has order $\prod_{j=1}^i p_j$. Let $E_i=E^{G_i}$ be the fixed subfield under $G_i$. By the Galois theory \cite{hir08}, we have a tower of fields
\begin{equation}\label{eq:ft}
F=E_r\subsetneq E_{r-1}\subsetneq \ldots \subsetneq E_1\subset E_0=E.
\end{equation}
Next, we construct an $\F_q$-basis $\cB$ of $\cL\big(\Gl P_{\infty}\big)$ from the tower (see Lemma~\ref{lem:bases}). Under this basis $\cB$, a function $f\in\cL\big(\Gl P_{\infty}\big)$ can be written as a combination of $p_1$ functions in the Riemann-Roch space $\cL\big(\lfloor \Gl/p_1\rfloor P_{\infty}^{(1)}\big)$ of $E_1$ (where $P_{\infty}^{(1)}=P_{\infty}\cap E_1$ is a place of $E_1$). 
Thus,  MPE of $f\in\cL\big(\Gl P_{\infty}\big)$ can be reduced to $p_1$ MPEs of functions in $\cL\big(\lfloor \Gl/p_1\rfloor P^{(1)}_{\infty}\big)$ at $\calP^{(1)}=\calP\cap E_1$, which has size $|\calP^{(1)}|=N/{p_1}$, and  {then get $\ev_{\calP}(f)$ from these $p_1$ MPEs in $p_1O(N)$ operations (see Equation~\eqref{eq:fff} in \S 3)}. Denote the complexity of computing the MPE of length $N$ by $\fC(\lambda, N)$. By the above analysis, $\fC(\lambda, N)$ satisfies a recursive formula
\begin{equation}\label{eq:recfor}
\fC(\lambda, N)=p_1\fC\left(\lfloor \lambda/{p_1}\rfloor, N/{p_1}\right)+p_1O(N).
\end{equation}
Do the recursive reductions till to $E_r=F$, then $f$ can be written as a combination of $m$ functions in the Riemann-Roch space $\cL\big(\lfloor \Gl/m\rfloor  P_{\infty}^{(r)}\big)$ of $F=E_r$ (where $P_{\infty}^{(r)}=P_{\infty}\cap E_r=Q_{\infty}$). In this case, the recursive formula~\eqref{eq:recfor} leads the total running time equal to
\[\begin{split}
\fC(\lambda, N)&=p_1p_2\cdots p_r\cdot \fC\left(\lfloor \lambda/{m}\rfloor, s\right)+(p_1+p_2+\cdots +p_r)O(N)\\
&=m\cdot O(s\log s)+O(N\log m)=O(N\log N).
\end{split}\]
where the second equality follows from the assumption that the MPE of length $s$ of $F$ runs in $O(s\log s)$ operations.

\subsection{Organization of the paper}
The paper is organized as follows. In Section 2, we introduce some facts about function fields including function field extensions and a tower of function fields. Algebraic geometry codes are also introduced in Section 2. Section 3 presents our divide-and-conquer technique by reducing MPE on an extension field $E$ to a subfield $F$ and concludes our main result on encoding AG codes from plane curves. Section 4 instantiates some specific field extensions that satisfy all conditions proposed in Section 3, i.e., the Kummer extension, the Artin-Schreier extension, and the mixed extension of Kummer and Artin-Schreier. In Section 5, we give examples of algebraic geometry codes that are encodable with complexity $O(N\log N)$. Section 6 presents an encoding algorithm for algebraic geometry codes from the Hermitian tower with complexity $O(N\log N)$.

\section{Preliminary}
In this section, let us review some results on function fields.
\subsection{Function fields}\label{sec:ff}
Let us introduce some basic facts about function fields over finite fields. The reader may refer to \cite{C51,sti09} for the details.

Let $F$ be a field containing a finite field $\F_q$. We say that $F$ is a function field of one variable over $\F_q$ if there is an element $x\in F$ that is transcendental over $\F_q$ such that the field extension $F/\F_q(x)$ is algebraic. We say that $\F_q$ is the full constant field of $F$ if every element of $F\setminus\F_q$ is transcendental over $\F_q$. In this case, we simply denote it by $F/\F_q$.
Let $F/\F_q$ be a function field of one variable over $\F_q$. Let $\PP_F$ denote the set of places of $F$. For each place $P\in \PP_F$, one can define a discrete valuation $\nu_P$ which is a map   from $F$ to $\ZZ\cup\{\infty\}$ satisfying certain properties (see \cite[Chapter 1]{sti09}). The integral  ring of $P$ is given by ${\calO}_P:=\{x\in F:\; \nu_P(x)\ge 0\}$. Then ${\calO}_P$ is a local ring and its unique maximal ideal is $\{z\in F:\; \nu_P(z)>0\}$. With abuse of notation, we still denote this ideal by $P$. The ideal $P$ is a principal ideal. A generator $t$ of $P$ is called a prime element or local parameter at $P$. It is easy to see that $t\in F$ is a local parameter at $P$ if and only if $\nu_P(t)=1$.  The residue class field ${\calO}_P/P$ is denoted by $F_P$.  Then $F_P$ is a field extension of $\F_q$. The degree of $P$, denoted by $\deg(P)$, is defined to be the extension degree $[F_P:\F_q]$. We say that $P$ is $\F_q$-rational (or simply rational if there is no confusion) if $\deg(P)=1$.
The free abelian group generated by $\PP_F$ is called the divisor group of $F$, denoted by $\Div(F)$. Thus, every element of $\Div(F)$ has the form $\sum_{P\in\PP_F}n_PP$, where $n_P\in\ZZ$ and only finitely many $n_P$ are nonzero. We say that a divisor $D=\sum_{P\in\PP_F}n_PP $ is effective or positive if $n_P\ge 0$ for all $P\in\PP_F$. We use the notation $D\ge 0$ for an effective divisor $D$. The degree of $D=\sum_{P\in\PP_F}n_PP $, denoted by $\deg(D)$, is defined to be $\deg(D)=\sum_{P\in\PP_F}n_P\deg(P)\in \ZZ $. For a nonzero element $z\in F$, we define its principal divisor by $\ddiv(z)=\sum_{P\in\PP_F}\nu_P(z)P$ and  its pole divisor by $(z)_{\infty}=-\sum_{P: \nu_P(z)<0}\nu_P(z)P$.

For a divisor  $D=\sum_{P\in\PP_F}n_PP $, one can associate it with an $\F_q$-vector space, called the Riemann-Roch space
\begin{equation}\label{eq:2.1.1} \cL(D)=\{z\in F\setminus\{0\}:\; \ddiv(z)+D\ge 0\}\cup\{0\}.\end{equation}
Then $\cL(D)$ is an $\F_q$-vector space of dimension at least $\deg(D)+1-g$. We denote this dimension by $\ell(D)$. Furthermore, we have the equality $\ell(D)=\deg(D)+1-g$ if $\deg(D)\ge 2g(F)-1$, where $g(F)$ is the genus of $F$.

For a function field $F/\F_q$, denote by $N(F)$ the number of $\F_q$-rational places. Then the well-known Hasse-Weil bound says that
\begin{equation}\label{eq:2.1.2} N(F)\le q+1+2g(F)\sqrt{q}.\end{equation}
An algebraic function field achieving the above Hasse-Weil bound is called a maximal function field (or equivalently a maximal curve).

\subsection{Weierstrass semigroups}
Let $F/\F_q$ be a function field of one variable over $\F_q$ of genus $g$. Let $Q_\infty$ be a rational place of $F$.  If there exists an element $x\in F$ with pole divisor $(x)_\infty=nQ_\infty$ for a non-negative integer $n$, then such an integer $n$ is called a pole number of $Q_\infty$. The set of all pole numbers of $Q_\infty$ forms a numerical semigroup under addition, which is called the Weierstrass semigroup of $Q_\infty$ and denoted by $H(Q_\infty)$, i.e.,
\[H(Q_\infty)=\{n\in \mathbb{N}: \exists\ x\in F \text{ such that } (x)_\infty=nQ_\infty\}.\]
From the Weierstrass Gap Theorem \cite[Theorem 1.6.8]{sti09}, the cardinality of the set $\mathbb{N}\setminus H(Q_\infty)$ is $|\mathbb{N}\setminus H(Q_\infty)|=g$. 
For any non-negative integer $j\in H(Q_\infty)$, there exists an element $t_j\in F$ such that $(t_j)_\infty=jQ_\infty$. 
From the strict triangle inequality \cite[Lemma 1.1.11]{sti09}, these elements in $\{t_j: j\in H(Q_\infty)\}$ are linearly independent over $\F_q$, since they have pairwise distinct discrete valuations at the place $Q_\infty$. 
Let $\mu$ be a positive integer. The Riemann-Roch space $\cL(\mu Q_\infty)$ is an $\mathbb{F}_q$-vector space spanned by the elements in 
$$\{t_j: j\in H(Q_\infty) \text{ and } j\le \mu\}.$$

\subsection{Algebraic geometry codes}
Let $F/\F_q$ be a function field of one variable over $\F_q$ of genus $g$ with
$N(F)\ge 1$. Choose distinct rational places $P_{1},P_2,\ldots,P_N$ of $F$, where
$N>g$, and let $G$ be a divisor of $F$ with \supp$(G)\cap\{P_{1},P_2,\ldots,P_{N}\}
=\emptyset$. Consider the Riemann-Roch space $\cL(G)$ and note that
$\nu_{P_i}(f)\ge 0$ for $1\le i\le N$ and all $f\in\cL(G)$, i.e.,
$$\cL(G)\subseteq\bigcap_{i=1}^{N}{\calO}_{P_{i}}.$$
Thus, it is meaningful to define the $\F_q$-linear map
$\ev:\cL(G)\longrightarrow\F_{q}^
{N}$ by
$$\ev(f)=(f(P_{1}),f(P_2),\ldots,f(P_{N})) \qquad \hbox{for all} \; f\in\cL(G),$$
where $f(P)$ denotes, as usual, the residue class of $f\in {\calO}_P$ modulo the
place $P$. The image of $\ev$ is a linear subspace of $\F_{q}^N$ which is
denoted by $C(\cP, G)$ and called an {algebraic geometry code} (or {\bf
AG code}), where $\cP=\{P_1,P_2,\dots,P_N\}$. The map $\ev(f)$ is usually called the multipoint evaluation (MPE) of $f$ at $\calP$ and $N=|\calP|$ is called its length. To emphasize the evaluation is defined on the set $\calP$, we denote the map $\ev$ by $\ev_{\calP}$.

A fast algorithm for MPE is in fact a fast encoding algorithm for algebraic geometry codes. In this work, we mainly concern encoding algorithms for one-point algebraic geometry codes with complexity $O(N\log N)$, i.e., the one-point MPE problem on $F$ can be solved in $O(N\log N)$.

We wrap up this subsection by including some standard results on the parameters of algebraic geometry codes (see \cite[Chapter 2]{sti09}).
\begin{lemma}\label{lem:agcp}
\label{lem:2.1} { Let} $F/\F_q$ { be a function field of one variable of genus g
and let} $P_{1}, P_2, \ldots, P_N$ { be distinct rational places of F. Put $\cP=\{P_1,P_2,\dots,P_N\}$. Choose a
divisor G of F with} $g\le\deg(G)<N$ {\it and} $\supp(G)\cap\{P_{1},\ldots,
P_{N}\}=\emptyset$. {\it Then} $C(\cP, G)$ {\it is an}
$[N,k,d]$-{\it linear code over} $\F_q$ { with}
$$k=\ell(G)\ge\deg(G)-g+1, \qquad d\ge N-\deg(G).$$
{\it Moreover}, $k= \deg(G)-g+1$ {\it if} $\deg(G)\ge 2g-1$.
\end{lemma}

From Lemma \ref{lem:2.1}, we know that a function field $F/\F_q$ of genus $g$ with $N$ rational places gives a $q$-ary $[N,k,d]$-linear code $C(\cP, G)$ with
\begin{equation}\label{eq:2.2.1}R+\delta\ge 1-\frac{g-1}N,\end{equation}
where $R=\frac kn$ denotes the rate of the code $C(\cP, G)$ and $\delta$ denotes the relative minimum distance of the code $C(\cP, G)$.

\subsection{Extension of function fields}
Let both $E$ and $F$ be function fields of one variable with the full constant field $\F_q$ such that $E/F$ is a finite field extension. If $F$ and $E$ are the function fields of the curve ${\mathcal X}$ and ${\mathcal Y}$, respectively, we simply say that  ${\mathcal X}$ is a covering of ${\mathcal Y}$, i.e., there is a rational map from ${\mathcal Y}$ to ${\mathcal X}$.

 For simplicity, we always assume that $E/F$ is a Galois extension in this paper. Let $m=[E:F]$. Then for each place $Q$ of $F$, we have a factorization of $Q=\prod_{i=1}^rP_i^e$, where $P_1,P_2,\dots,P_r$ are places of $E$. The power $e$ is called the ramification index, denoted by $e(P_i|Q)$. Moreover, we have $\deg(P_1)=\deg(P_2)=\cdots=\deg(P_r)$ and the identity $m=re(P_i|P)f(P_i|Q)$, where $f(P_i|Q)=[E_{P_i}:F_Q]$ and it is called the relative degree of $P_i$ over $Q$. We say that (i) $Q$ is completely splitting if $r=m$ and $e(P_i|Q)=f(P_i|Q)=1$;  (ii) $Q$ is totally ramified if $r=f(P_i|Q)=1$ and $e(P_i|Q)=m$. In the factorization $Q=\prod_{i=1}^rP_i^e$, the places $P_i$ contains the place $Q$. Furthermore, a place $R$ of $E$ contains $Q$ if and only if $R$ appears in this factorization. We say that a place $P$ of $E$ lies over a place $Q$ of $F$ or $Q$ lies under $P$ if $P$ contains $Q$, denoted by $P|Q$. For a place $P$ of $E$ lies over a place $Q$ of $F$ with $P|Q$, apart from the above ramification index and relative degree, we have another important parameter called different exponent, denoted by $d(P|Q)$. The parameter $d(P|Q)$ is closely related to the ramification index $e(P|Q)$. More precisely speaking, we have $d(P|Q)\ge e(P|Q)-1$ and the equality holds if $\gcd(e(P|Q),q)=1$.

The Hurwitz genus formula tells us that the genus $g(E)$ of $E$ is determined by the genus $g(F)$ of $F$ and the different of $E/F$. Namely, we have
\[2g(E)-2=[E:F](2g(F)-2)+\sum_{Q\in\PP_F}\sum_{P|Q}d(P|Q)\deg(P).\]
To distinguish Riemann-Roch spaces of $E$ and $F$, we use $\cL_E(G)$ and $\cL_F(D)$ to denote the associated Riemann-Roch spaces of $G\in\Div(E)$ and $D\in\Div(F)$.

Let $E/F$ be a finite Galois extension. Let $Q$ be a place of $F$. The integral closure ${\calO}_Q(E)$ at $Q$ is defined to be the set of elements of $E$ that are integral over ${\calO}_Q$. It is proved in \cite[Corollary 3.3.5]{sti09} that  the integral closure ${\calO}_Q(E)$ at $Q$ is equal to $\bigcap_{P|Q} {\calO}_P$. A set $\{z_1,z_2,\dots,z_m\}$ of ${\calO}_Q(E)$ is called an integral basis at $Q$ if ${\calO}_Q(E)=\oplus_{i=1}^mz_i\calO_Q$.

\section{Encoding of algebraic geometry codes  }\label{sec:FMPE}

Let $E/\F_q$  be a function field of one variable over $\F_q$. Let $P_{\infty}, P_1,\dots,P_N$ be $N+1$ distinct rational places of $E$. Let us consider a one-point algebraic geometry code $C(\cP,\Gl P_\infty)$, where $\cP=\{P_1,P_2,\dots,P_N\}$ and $\Gl$ is a positive integer less than $N$. The main purpose of this section is to design an encoding algorithm for the code $C(\cP,\Gl P_\infty)$ with complexity $O(N\log N)$, or equivalently,  given an explicit basis of  $\cL(\Gl P_\infty)$, an FMPE algorithm for any $f\in \cL(\Gl P_\infty)$ at $\cP$.

Our idea is to reduce the MPE problem from the code $C(\cP,\Gl P_\infty)$ to the MPE problem from a shorter algebraic geometry code based on a subfield $F$ of $E$. We require that the extension $E/F$ satisfies the following properties to achieve this goal.

\bigskip
\fbox{\begin{minipage}{35em}
\begin{center} {\bf Four Properties for the extension $E/F$ (P4)}\end{center}
\begin{itemize}
\item[(i)] $E/F$ is an abelian extension with $E=F(y)$ of degree $m$ for some $y\in E$.
\item[(ii)] Assume all conjugate roots of $y$ are $\F_q$-linear functions of $y$, i.e., for all $\sigma\in \Gal(E/F)$,  $ \sigma(y)=ay+b$ for some $ a\in\F_q^*$ and $ b\in \F_q$.
\item[(iii)] There is a rational place $P_{\infty}$ of $E$ that is totally ramified in $E/F$. 
\item[(iv)] $\gcd(m,\nu_{\Pin}(y))=1$ and $\{1,y,\dots,y^{m-1}\}$ is an integral basis of $E/F$ at place $Q$ with $\Pin\nmid Q$. 
\end{itemize}
\end{minipage}}

\bigskip

Assume the degree of the field extension $E/F$ is $[E: F]=m$ with prime factorization $m=\prod_{i=1}^k p_i$. We say that $m$ is $B$-smooth for a positive real $B$ if $p_i\le B$ for all $1\le i\le r$.  We simply say that $m$ is smooth if all $p_i$ are constants.

Let $G=\Gal(E/F)$. By the first property in the above box (P4) and the structure theorem for finite abelian groups \cite{KS04}, then there is an ascending chain of subgroups of $G$ as follows
\[\{1\}=G_0\subsetneq G_{1}\subsetneq \dots \subsetneq G_{r-1}\subsetneq G_r=G,\ \text{where}\ p_{i}=|G_{i}/G_{i-1}|\ \text{for\ all}\ 1\le i\le r.\]
Let $E_k=E^{G_{k}}$ be the fixed subfield of $E$ under $G_k$, i.e.,
\[E_k=\{u\in E :\; \sigma(u)=u\ \text{for\ all}\ \sigma\in G_{k}\}.\]
Then $E_0=E$ and $E_r=F$.  Furthermore, $E_0/E_k$ is a Galois extension with degree $[E_0: E_k]=|G_{k}|=\prod_{i=1}^k p_i$ for all $1\le k\le r$.  Thus, we obtain a tower of fields
\[E=E_0\supsetneq E_{1}\supsetneq \dots \supsetneq E_{r-1}\supsetneq E_r=F.\]
As $E/F$ is a finite Galois extension, the extension $E/F$ is simple, i.e., there exists an element $y\in E$ such that $E=F(y)$.
Define
\begin{equation}\label{eq:doy}
y_k=\prod_{\sigma\in G_k} \sigma(y), \ k=0,1,\dots,r.\end{equation}
Then, by the second property in the box (P4),
 \begin{equation}\label{eq:royk}
 y_k=\phi_{k}(y),\ \text{for\ some\ polynomial}\ \phi_k(T)\in \F_q[T]\ \text{of\ degree}\ |G_k|.
 \end{equation}
This implies that $[E: F(y_k)]\leq |G_k|$. Now  we claim that $ E_k= F(y_k)$ for all $0\le k\le r$. 
From Galois theory, it is easy to see that $y_k\in E_k$ and $F(y_k)\subseteq E_k$. Moreover, we have $|G_k|=[E: E_k]\leq [E: F(y_k)]\leq |G_k|$ which forces $[E: F(y_k)]=[E: E_k]$. 
Note that $y_0=y$ and $y_r\in F$.

For any integer $i$ satisfying $0\leq i<m=\prod_{k=1}^r p_k$, it can be uniquely written as
\begin{equation}\label{eq:exp}
i=i_0+i_1\cdot p_1+i_2\cdot p_1p_2+\cdots+i_{r-1}\cdot \prod_{j=1}^{r-1} p_j,\ \text{where\ each}\ 0\leq i_k\leq p_{k+1}-1.
\end{equation}
Let us denote by the boldface $\bfi$ the vector $(i_{r-1},i_{r-2},\dots,i_0)$ given in \eqref{eq:exp}.
Furthermore, we define a function $\by^{\bfi}$ associated with $i$ as follows
\begin{equation}\label{eq:expansion}
\by^{\bfi}:= y_{r-1}^{i_{r-1}}\cdot y_{r-2}^{i_{r-2}}\cdots y_0^{i_0}.
\end{equation}

The following lemma shows that an $\F_q$-basis of the Riemann-Roch space $\cL(\lambda P_{\infty})$ can be constructed explicitly if the extension $E/F$ satisfies four properties given in the above box (P4). Let $Q_\infty$ be the unique place of $F$ lying under $P_\infty$, i.e., $Q_\infty=P_\infty\cap F$.

\begin{lemma}\label{lem:bases}
Let $E=F(y)$ be a Galois extension satisfying (P4) given in the above box. Assume $d=-\nu_{P_{\infty}}(y)$. Let $y_0,y_1,\dots,y_r$ be defined as in Equation~\eqref{eq:doy} and ${\bf i}=(i_{r-1},i_{r-2},\dots,i_0)$ defined by Equation~\eqref{eq:exp} for $i=0, 1,\dots, m-1$. For any integer $j\in H(Q_\infty)$ in the Weierstrass semigroup of $Q_{\infty}$, let $t_j$ be an element in $F$ such that its pole divisor $(t_j)_{\infty}=jQ_{\infty}$.  For an integer $\lambda>0$, set
\[\cB_i=\left\{\by^{\bf i}\cdot t_j:=y_{r-1}^{i_{r-1}}\cdots y_1^{i_1}y_0^{i_0}\cdot t_j \mid  i\cdot d+j\cdot m\leq \lambda\ \text{and}\ j\in H(Q_\infty)\right\},\ i=0,1,\ldots,m-1.\]
If $\gcd(d, m)=1$ and $\{1, y, \dots, y^{m-1}\}$ is an integral basis of $E/F$ at $Q$ for any place $Q\neq Q_{\infty}$ of $F$, then $\cB=\cup_{i=0}^{m-1}\cB_i$ is an $\F_q$-basis of the Riemann-Roch space $\cL(\lambda P_{\infty})$.
 \end{lemma}
 
\begin{proof} Firstly, for each $\by^{\bfi}\cdot t_j\in\cB_i$, we have
 \[\begin{split}
 \nu_{{P}_{\infty}}(\by^{\bfi}\cdot t_j)&=\sum_{k=0}^{r-1} i_k\cdot\nu_{P_{\infty}}(y_k)+\nu_{P_{\infty}}(t_j)\\
 &=\left(\sum_{k=0}^{r-1} i_k\cdot |G_k|\right)\cdot \nu_{P_{\infty}}(y)+  \nu_{Q_{\infty}}(t_j)\cdot e(P_{\infty}|Q_{\infty})\\
 &=-(i\cdot d+j\cdot m)\geq -\lambda,\end{split}\]
where the first equality follows from Equation~\eqref{eq:expansion}, the second equality follows from the Equation~\eqref{eq:royk}, and the third equality follows from Equation~\eqref{eq:exp} and the fact that $P_{\infty}$ is totally ramified in $E/F$. Then $\cB_i\subseteq \cL(\lambda P_{\infty})$ and hence $\cB\subseteq \cL(\lambda P_{\infty})$. Moreover, we claim that
\begin{equation}\label{eq:bcond}
 \nu_{{P}_{\infty}}(\by^{\bfi}\cdot t_j)\neq \nu_{{P}_{\infty}}(\by^{\bfi'}\cdot t_{j'}),\ \text{if}\ 0\leq i\neq i'\leq m-1\ \text{or}\ j\neq j'.\end{equation}
 Otherwise, one would have
 \[i\cdot d+j\cdot m=i'\cdot d+j'\cdot m,\ \text{i.e.,}\ (i-i')\cdot d=(j'-j)\cdot m.\]
As $\gcd(d, m)=1$ by assumption, then $m\mid (i-i')$  contradicts the assumption $0\leq i\neq i'\leq m-1$. Thus all elements in $\cB$ are $\F_q$-linearly independent.

Next, we will show that the set $\cB$ spans the whole Riemann-Roch space $\cL(\lambda P_{\infty})$. For any $f\in \cL(\lambda P_{\infty})$, assume $f=\sum_{i=0}^{m-1}\alpha_iy^i$, where $\alpha_i\in F$ for all $0\le i\le m-1$. Since $\{1, y, \dots, y^{m-1}\}$ is an integral basis of $E/F$ at $Q$ for any place $Q\neq Q_{\infty}$ of $F$, by \cite[Corollary 3.3.5]{sti09}, one has
\[\bigcap_{P\mid Q}\calO_{P}=\bigoplus\limits_{i=0}^{m-1} \calO_{Q}y^i,\ \forall\  Q\neq Q_{\infty}.\]
For any $Q\neq Q_{\infty}$, $f\in\bigcap_{P\mid Q}\calO_{{P}}$. By the above equation, the coefficients $\alpha_0,\dots,\alpha_{m-1}$ of $f$ are all in $\calO_{Q}$, namely, they have poles only at $Q_{\infty}$. Assume $\nu_{Q_\infty}(\alpha_i)=-a_i$ for $0\le i\le m-1$. Then $\alpha_i\in \cL(a_iQ_\infty)$. This means that $\Ga_i$ can be expressed as an $\F_q$-linear combination of elements in $\{t_j: 0\le j\le a_i, j\in H(Q_\infty)\}$. In particular, $a_i\in H(Q_\infty)$ is a pole number and the pole divisor of $\alpha_i$ in $F$ is $(\alpha_i)_\infty^F=a_iQ_\infty$. By the strict triangle inequality \cite[Lemma 1.1.11]{sti09}, we have 
\[\nu_{P_{\infty}} (f)=\min_{0\leq i\leq m-1}\big\{\nu_{P_{\infty}}\left(\alpha_iy^{i}\right)\big\}=-\max_{0\leq i\leq m-1} \big\{a_i\cdot m+i\cdot d\big\}\geq -\lambda.\]
That is to say, $a_i\cdot m+i\cdot d\leq \lambda$ for all $0\le i\le m-1$. Furthermore, $y^i$ can be written as an $\F_q$-linear combination of elements in $\{\by^{\bfi^\prime}: 0\le i^\prime \le i\}$ for any $0\le i\le m-1$, since $\by^{\bfi^\prime}$ is a polynomial in the variable $y$ of degree $i^\prime$. It follows that $\alpha_iy^i$ is an $\F_q$-linear combination of elements in $\cup_{i^\prime=0}^{i}\cB_{i^\prime}$ for each $0\le i\le m-1$.
Thus, $f=\sum_{i=0}^{m-1}\alpha_iy^i$ is an $\F_q$-linear combination of elements in $\cup_{i=0}^{m-1}\cB_i$. This completes the proof.
\end{proof}

\begin{remark}
If $F$ is the rational function field over $\F_q$, then there exists an element $x\in F$ such that $(x)_\infty=Q_\infty$. Hence, we have $F=\F_q(x)$ and any non-negative integer is a pole number of $Q_\infty$. For simplicity, $t_j$ can be chosen as $x^j$ for any $j\in H(Q_\infty)=\mathbb{N}$. 
\end{remark}

Let $E/F$ be an abelian extension of degree $m$ satisfying (P4) given in the above box. Let $Q_{\infty}=P_{\infty}\cap F$ be the place lying under $P_{\infty}$. Suppose there is a set $\cQ\subsetneq \PP_{F}$ satisfying all places in $\cQ$ are splitting completely in $E/F$. Assume $|\cQ|=s$ and $\cQ=\{Q_1,Q_2,\dots,Q_s\}$. Then the set $\cP=\{P\in\PP_E:\; P\mid Q_i,\; 1\le i\le s\}$ has cardinality $N=ms$. Label elements of $\cP$ by $\{P_1,P_2,\dots,P_N\}$.  Let $\lambda$  be a positive integer such that the dimension of $\cL(\lambda P_{\infty})$ is at most $N$. We now have two algebraic geometry codes $C(\cP,\Gl P_\infty)$ from the function field $E$ and $C(\cQ,\lfloor \lambda/m\rfloor Q_\infty)$ from the function field $F$, respectively. We will design an FMPE Algorithm~1, namely the encoding algorithm for $C(\cP,\Gl P_\infty)$ based on an FMPE algorithm for $C(\cQ,\lfloor \Gl/m\rfloor Q_\infty)$.

 \begin{algorithm}[!h]\label{alg: FMPE}
  \caption{ FMPE($f,\calP$).}
  \label{alg:Framwork}
  \begin{algorithmic}[1]
    \Require
     $f\in\cL(\lambda P_{\infty})$ and $\calP=\{P_{1},P_2,\dots,P_{N}\}$.
    \Ensure
      $\ev_{\calP}(f)=\left(f(P_{1}),\dots,f(P_{N})\right)\in\F_q^N$.
    \State
Write $f=\sum_{i=0}^{m-1}a_i(t)\by^{\bfi}$.
    \State For $j=0,\dots,r$, let $\calP^{(j)}=\calP\cap E_j$. In particular, $\calP^{(0)}=\calP$.
    \For {$j=0$ to $r-1$}
     \State $j\gets j+1$
     \State Write $f=f_{j,0}+f_{j,1}y_{j-1}+\dots+f_{j,p_j-1}y_{j-1}^{p_j-1}$.
    \If {$j=r$}
\State Compute $\ev_{\calP^{(j)}}(f_{j,k})$ for $k=0,\dots, p_j-1$.
\Else
     \If {$j<r$}
     \State Call FMPE$\left(f_{j,k},\calP^{(j)}\right)$ to compute $\ev_{\calP^{(j)}}(f_{j,k})$ for $k=0,\dots, p_j-1$.
     \EndIf
     \EndIf
     \\
     \Return  $\ev_{\calP^{(j-1)}}(f)[P]=f_{j,0}(P)+f_{j,1}(P)y_{j-1}(P)+\dots+f_{j,p_j-1}(P)y_{j-1}^{p_j-1}(P)$ for every $P\in \calP^{(j-1)}$.
     \EndFor
  \end{algorithmic}
\end{algorithm}

Let us show the correctness of the FMPE Algorithm~1 and analyze its complexity.
\begin{theorem}\label{thm:main} Assume that the extension $E/F$ satisfies the four properties given in the box (P4)  and $m$ has prime factorization $\prod_{i=1}^rp_i$ with $p_i\le B$. Let the notions $Q_{\infty}$, $\cQ$, $\cP$, and $\lambda$ be given as above (here we assume all places in $\cQ$ are splitting completely in $E/F$).
If the MPE of the algebraic geometry code $C(\cQ,\lfloor \lambda/m\rfloor Q_\infty)$ can be done in time $O(B s\log s)$, then the encoding algorithm for $C(\cP,\Gl P_\infty)$, namely, the FMPE Algorithm 1, can run in time $O(B\cdot N\log N)$.
\end{theorem}
\begin{proof}
Our assumptions  ensure that  $\cL(\lambda P_{\infty})$ has an $\F_q$-basis given as in Lemma~\ref{lem:bases}. Thus, any function $f\in\cL(\lambda P_{\infty})$ can be written as
\[f=\sum_{i=0}^{m-1}c_i(\bt)\cdot y_{r-1}^{i_{r-1}}\cdots y_1^{i_1}y_0^{i_0},\]
where $i=\sum_{k=0}^{r-1}i_k\cdot|G_k|$ and $c_i(\bt)=\sum_{j\in H(Q_\infty)} c_{i,j}t_j$ is a finite sum with $c_{i,j}\in\F_q$ satisfying   
\begin{equation}
-\nu_{Q_{\infty}}\left(c_i(\bt)\right)\cdot m+i\cdot d\leq \lambda\ \text{for\ each}\ 0\le i\le m-1.
\end{equation}

By combining all the terms of $f$ which have the same divisor $y_0^{i_0}$ for each $i_0=0,1,\ldots,p_1-1$, which can be done in at most $O(\ell(\lambda P_{\infty}))\leq O(N)$ steps, we assume
\begin{equation}\label{eq:recs}
f=f_0(\bt, y_1)+f_1(\bt, y_1)\cdot y+\cdots+f_{p_1-1}(\bt, y_1)y^{p_1-1},
\end{equation}
where the symbol $\bt$ denotes the multivariate of elements in $\{t_j: j\in H(Q_{\infty}),\ j\leq \lambda/m\}$ and $f_0, f_1, \dots, f_{p_1-1}\in \F_q[\bt, y_1]$.
Let $P_{\infty}^{(1)}$ be the unique place of $E_1$ lying under $P_{\infty}$. Then, for every $k=0, 1,\dots, p_1-1$,
\[\begin{split}
\nu_{P_{\infty}^{(1)}}\left(f_k(\bt, y_1)\right)&=\nu_{P_{\infty}}\left(f_k(\bt, y_1)\right)/p_1\geq \frac 1{p_1}\cdot\nu_{P_{\infty}}\left(f_k(\bt, y_1)y^k\right)\\
&\geq \frac 1{p_1}\cdot\min_{i:\ i_0=k}\{ \nu_{P_{\infty}}(c_i(\bt)\by^{\bfi})\}\\
& {\geq \frac 1{p_1}\cdot\min_{i}\{ \nu_{P_{\infty}}(c_i(\bt)\by^{\bfi})=  \nu_{P_{\infty}}(f)/p_1 \geq -\lambda/p_1.}
\end{split}
 \]
Thus, $f_0, f_1,\dots, f_{p_1-1}\in \cL\left(\lfloor \lambda/{p_1}\rfloor P_{\infty}^{(1)}\right)$. 

For any $Q_j\in \cQ$, let $\cP_j=\{P_{j,1}, P_{j,2}, \cdots, P_{j,m}\}$ be the set of $m$ rational places in $E$ lying above $Q_j$. Then
\[\cP=\cup_{j=1}^s \cP_j=\{P_{1,1}, \dots, P_{1,m},\dots,P_{s,1}, \dots, P_{s,m}\}.\]
Note that, for each $P_{j, \ell}\in\calP_j$, the evaluations $f_{k}(P_{j, \ell})=f_k(P_{j, \ell}\cap E_1)$ for all $f_k$.
Let
\[\calP^{(1)}_j=\{P_{j, \ell}\cap E_1:\; \text{for\ all}\ P_{j, \ell}\in\calP\}.\] 
Then $|\calP^{(1)}_j|=m/p_1$ and the set $\calP^{(1)}=\cup_{j=1}^{s}\calP^{(1)}_j$ has cardinality $s\cdot m/p_1$. By equation~\eqref{eq:recs}, the evaluation $\ev_{\cP}(f)$ of $f$ at $\cP$ can be reduced to the evaluations $\ev_{\calP^{(1)}}(f_0), \ev_{\calP^{(1)}}(f_1), \dots,$ $\ev_{\calP^{(1)}}(f_{p_1-1})$ of $f_0, f_1, \dots, f_{p_1-1}$ at $\calP^{(1)}$ and then get $f(P)$ by 
\begin{equation}\label{eq:fff}
 {f(P)=f_0(P)+f_1(P)y(P)+\ldots +f_{p_1-1}(P)y(P)^{p_1-1}.}
\end{equation}
This reduction corresponds to the steps $5-13$ in our FMPE Algorithm~1. We continue the reduction in this fashion till to $E_r=F$. 
Let $P_{\infty}^{(i)}=P_\infty\cap E_i$ for any $0\le i\le r$. Particularly, $P_{\infty}^{(r)}=Q_\infty$.
At the last step in the recursion, we are facing computations of $s$-point evaluations of $m$ functions in $\cL(\lfloor \lambda/m\rfloor Q_{\infty})$ at the set $\cQ\subseteq \mathbb{P}_F$. By assumption, we can directly compute the $\ev_{\cQ}(g)$ for any $g\in \cL(\lfloor \lambda/m\rfloor Q_{\infty})$. Therefore, we have proved the correctness of Algorithm~1 FMPE.

Next, we analyze the complexity. Denote the complexity (i.e., the total number of operations in $\F_q$) of evaluating an $f\in\cL(\lambda P_{\infty})$ at the set $\cP$ by $\fC(\lambda, N)$. By the above discussion, the evaluation $\ev_{\cP}(f)$ can be reduced to the $p_1$ evaluations of functions in $\cL\left(\lfloor \lambda/{p_1}\rfloor P_{\infty}^{(1)}\right)$ of length $N/{p_1}$.  {Then, for every $P\in\calP$, $f(P)$ is a combination of these values by Equation \eqref{eq:fff} which will cost $p_1$ operations in $\F_q$}. Thus, $\fC(\lambda, N)$ satisfies the following recursive formula
\begin{equation}\label{eq:time}
\fC(\lambda, N)=p_1\fC(\lfloor \lambda/{p_1}\rfloor, N/{p_1})+p_1\cdot O(N).
\end{equation}
At the last step of the recursive reduction, this formula \eqref{eq:time} leads the total running time equal to
\[\begin{split}
\fC(\lambda, N)&=m\cdot \fC(\lfloor \lambda/{m}\rfloor, s)+(p_1+\cdots +p_r)O(N)\\
&=m\cdot O(B\cdot s\log s)+O(B\cdot N\log m)=O(B\cdot N\log N).
\end{split}\]
where the second equality follows from the assumption that encoding of  $C(\cQ,\lfloor \lambda/m\rfloor Q_\infty)$ runs in $O(B s\log s)$ operations.
\end{proof}

\begin{remark}[Inverse FMPE]
 {
In the univariate case, it is known that both the FFT and inverse FFT of length $s$ can be implemented in $O(s\log s)$ operations of the underlying finite field $\F_q$. Actually, under the assumptions in Theorem~\ref{thm:main}, we can show that there exists an inverse FMPE algorithm that can compute an $f\in\cL(\lambda P_{\infty})$ from the MPE $\ev_{\calP}(f)\in\F_q^N$ in $O(N\log N)$ time. In \cite{BRS20}, the inverse FMPE is also called the unencoding problem. Thus, our results also improve complexity of the unecoding problem in \cite{BRS20}. We follow the notations defined in the proof of Theorem~\ref{thm:main}. Let $I(N)$ denote the complexity of inverse FMPE of length $N$. For any rational place $P^{(1)}\in\calP^{(1)}=\calP\cap E_1$, assume $P_1, P_2, \ldots, P_{p_1}$ are $p_1$ rational places lying over $P^{(1)}$. By the Equation~\eqref{eq:recs}, we have
\begin{equation}\label{eq:irec}
\left(\begin{matrix} f_0(P^{(1)})\\
f_1(P^{(1)})\\
\vdots\\
f_{p_1-1}(P^{(1)})\end{matrix}\right)=\left(\begin{matrix}
1            &            y(P_{1})    & \cdots  & y^{p_1-1}(P_{1})\\ 
1 & y(P_{2}) & \cdots  & y^{p_1-1}(P_{2}) \\
\vdots       & \vdots       & \cdots   & \vdots \\
1 &   y(P_{p_{1}}) & \cdots &  y^{p_1-1}(P_{p_{1}})\end{matrix}\right)^{-1}\cdot \left(\begin{matrix} f(P_{1})\\
f(P_{2})\\
\vdots\\
f(P_{p_1})\end{matrix}\right)\end{equation}
Thus, we can first compute the MPEs of $f_0, f_1,\ldots, f_{p_1-1}$ at $\calP^{(1)}$ from the above equation, which will cost $p_1O(N)$ operations. Then we can reduce the inverse FMPE of $\ev_{\calP}(f)$ to $p_1$ inverse FMPEs of $\ev_{\calP^{(1)}}(f_0)$, $\ev_{\calP^{(1)}}(f_1)$, $\ldots$, $\ev_{\calP^{(1)}}(f_{p_1-1})$. Finally, we get $f$ by Equation~\eqref{eq:recs} which will cost at most $O(N)$ operations. Therefore, $I(N)$ satisfies the following recursive formula
\[I(N)=p_1I(N/{p_1})+p_1O(N)=O(N\log N).\]}
\end{remark}

\section{Instantiations}\label{sec:inst}
In this Section, we instantiate the encoding algorithm for the general function field extension given in  Section 3  by some special extensions $E/F$ with $F=\F_q(x)$ and $E=\F_q(x,y)$. In this case, the FMPE of length $s$ for $F$ is just the fast Fourier transform of length $s$ over the finite field $\F_q$. By employing the known FFT from \cite{cooley1965, lin2014novel} with quasi-linear $O(s\log s)$ operations in $\F_q$, we can obtain an encoding algorithm of algebraic geometry code from $E$ with quasi-linear time $O(N\log N)$.
In the following, we will consider the two most common types of  Galois extension $E/F$, namely the Kummer extensions and Artin-Schreier extensions, as well as the mixed extension of Kummer and Artin-Schreier.

\subsection{Kummer extensions}
Let $\F_q$ be a finite field and $F=\F_q(x)$ be the rational function field with full constant field $\F_q$. Assume $\F_q$ contains a primitive $m$-th root of unity, denoted by $\omega_m$, for some positive integer $m$ satisfying $\gcd(m, q)=1$. Suppose $u(x)\in\F_q[x]$ is a square-free polynomial with degree $d=\deg\left(u(x)\right)$ relatively prime to $m$. Define $\varphi(T)=T^m-u(x)$. Let
\begin{equation}\label{eq:ke}
E=F(y)\ \text{with}\ \varphi(y)=0.
\end{equation}
Then $E$ is a Kummer extension of $F$ and we have the following proposition.

\begin{proposition}\label{prop:kp}
\begin{enumerate}
\item[{\rm (i)}] The extension $E/F$ is Galois of degree $m=[E: F]$ and its Galois group $\Gal(E/F)$ is isomorphic to the cyclic group $\langle \omega_m \rangle$; more precisely, any $F$-automorphism of $E$ is given by $\sigma(y)=\omega y$ for some $\omega\in \langle \omega_m \rangle$.
\item[{\rm (ii)}] The infinity rational place ${Q}_{\infty}$ of $F$, i.e., the pole of $x$, is totally ramified in $E/F$. Let ${P}_{\infty}$ be the unique place in $E$ lying above ${Q}_{\infty}$. Then the discrete valuations of $x, y$ at ${P}_{\infty}$ are
\[\nu_{{P}_{\infty}}(x)=-m,\ \nu_{{P}_{\infty}}(y)=-d.\]
\item[{\rm (iii)}] For any place ${Q}\neq {Q}_{\infty}$, $\{1, y,\dots, y^{m-1}\}$ is an integral basis of $E/F$ at ${Q}$.
\item[{\rm (iv)}] Assume ${Q}$ is a rational place of $F$ such that the evaluation of $u(x)$ at ${Q}$ is an $m$-th power, i.e., $u({Q})=\alpha^m$ for some $\alpha\in\F_q^*$. Then $Q$ splits completely in $E/F$.

\end{enumerate}
\end{proposition}
 \begin{proof}
See \cite[Proposition 3.7.3]{sti09} for the proof of (i) and (ii), while (iv) directly follows from \cite[Corollary 3.3.8 (c)]{sti09}. We only need to prove (iii).

It is clear that $\{1, y,\dots, y^{m-1}\}$ is a basis of $E$ as a vector space over $F$. 
For any $Q\neq Q_{\infty}$, let ${P}\in \PP_E$ be a place of $E$ lying above ${Q}$.

Case I: if $\nu_{{Q}}\left(u(x)\right)=0$, then $\nu_{{P}} (y)=0=\nu_{{P}}(\varphi'(y))$. Thus, by \cite[Corollary 3.5.11]{sti09}, $\{1,y,\dots, y^{m-1}\}$ is an integral basis for $E/F$ at ${Q}$.

Case II: if $\nu_{{Q}}\left(u(x)\right)>0$, then $\nu_{Q}\left(u(x)\right)=1$ as $u(x)$ is square-free. In this case, ${Q}$ is totally ramified in $E/F$ and $\nu_{{P}}(y)=1$. Thus, the different exponent \[d({P} \mid {Q})=e({P} \mid {Q})-1=m-1=\nu_{{P}}(\varphi'(y)).\]
 By \cite[Theorem 3.5.10]{sti09}, $\{1,y,\dots, y^{m-1}\}$ is an integral basis for $E/F$ at ${Q}$.
This completes the proof.
 \end{proof}

In conclusion, the Kummer extension $E/F$ of degree $m$ satisfies all four properties (P4). Next, we consider $m\mid (q-1)$ to be smooth. Since $q-1$ is smooth with high probability for odd $q$, we assume $2\mid (q-1)$ and $m$ has prime factorization $m=2\prod_{i=2}^rp_i$, i.e., $p_1=2$.
Consequently, for any positive integer $\Gl<N=ms$, one can firstly construct a basis of $\cL(\Gl P_{\infty})$ by Lemma~\ref{lem:bases} and then perform FMPE of any function in $\cL(\Gl P_{\infty})$. To be more precise, we can describe an $\F_q$-basis $\cB$ of $\cL(\Gl P_{\infty})$ explicitly. As $G=\Gal(E/F)$ is cyclic and the conjugates of $y$ are
\[\{\sigma(y):\ \sigma\in G\}=\{\omega\cdot y:\ \omega^m=1\}.\] We follow notations defined as in Section~\ref {sec:FMPE}. After computations, we have
 \begin{enumerate}
 \item $G_i\cong \langle \omega_m^{m/(2p_2\cdots p_i)} \rangle$, $i=1,\dots,r$;
 \item $E_i:=E^{G_i}=F(y_i)$, where $y_i:=-\prod_{\sigma\in G_i} \sigma(-y)=y^{|G_i|}$. Thus, $y_i=y_{i-1}^{p_i}$ for all $i$.
 \item The basis $\cB$ of $\cL(\Gl P_{\infty})$ is
 \begin{equation}\label{eq:kbases}
 \begin{split}
 \cB&=\{y_0^{i_0}y_1^{i_1}\cdots y_{r-1}^{i_{r-1}}x^j \mid i\cdot d+j\cdot m\leq \Gl,\ \text{and}\ 0\leq i\leq m-1,\ j\geq 0\}.\\
 &=\{y^ix^j \mid i\cdot d+j\cdot m\leq \Gl,\ \text{and}\ 0\leq i\leq m-1,\ j\geq 0\}.
 \end{split}
 \end{equation}
 \end{enumerate}


\begin{corollary}\label{cor:k}
Let $E/\F_q(x)$ be a Kummer extension of degree $m\mid (q-1)$ defined by equation~\eqref{eq:ke} and let ${P}_{\infty}$ be the unique pole  of $x$ in $\mathbb{P}_E$.  Let $\cQ=\{\Ga\in \F_q:\; x-\Ga \ \mbox{completely  splits  in $E/\F_q(x)$}\}$ and $s=|\cQ|$. Let $\cP$ be the set of all rational places in $E$ lying above places in $\cQ$. If $q-1$ is $O(1)$-smooth and  {the MPE of any polynomial in $\F_q[x]_{<s}$ at $\cQ$ can be done in $O(s\log s)$ operations}, then encoding $C(\cP,\lambda{P}_{\infty})$ with $\Gl<N=ms$ costs $O( N\log N)$ operations of $\F_q$.
\end{corollary}

\subsection{Artin-Schreier extensions}
Let $\F_q$ be a finite field with characteristic $p$. Let $W$ be an $\F_p$-subspace of $\F_q$ of dimension $r$. Assume $\{1=w_1, w_2, \dots, w_r\}$ is an $\F_p$-basis of $W$. Define
\[L(T)=\prod_{w\in W}(T-w).\]
Then $L(T)\in \F_q[T]$ is a $p$-linearized polynomial of degree $p^r$, namely, $L(a\alpha+b\beta)=aL(\alpha)+bL(\beta)$ for any $a,b\in\F_p$ and $\alpha, \beta\in \F_q$. Assume $u(x)\in \F_q[x]$ is a polynomial of degree $d$ coprime to $p$. Let $F=\F_q(x)$ be the rational function field. Assume \[E=F(y)\ \text{with}\ L(y)=u(x).\] Then $E/F$ is an Artin-Schreier extension that has the following properties:

\begin{proposition}\label{prop:at}
\begin{itemize}
\item[{\rm (i)}] The extension $E/F$ is Galois of degree $p^r=[E: F]$ and its Galois group $\Gal(E/F)$ is isomorphic to $(\ZZ/p\ZZ)^r$; more precisely, any $F$-automorphism of $E$ is given by $\sigma(y)= y+w$ for some $w\in W$.
\item[{\rm (ii)}] The pole ${Q}_{\infty}$ of $x$ in $F$ is totally ramified in $E/F$. Let ${P}_{\infty}$ be the unique place in $E$ lying above ${Q}_{\infty}$. Then the discrete valuations of $x, y$ at ${P}_{\infty}$ are
\[\nu_{{P}_{\infty}}(x)=-p^r,\ \nu_{{P}_{\infty}}(y)=-d.\]
\item[{\rm (iii)}] For any place ${Q}\neq {Q}_{\infty}$, $\{1, y,\dots, y^{p^r-1}\}$ is an integral basis of  $E/F$ at ${Q}$.
\item[{\rm (iv)}] Assume ${P}$ is a rational place of $F$ such that $L(T)=u({P})$ is solvable in $\F_q$. Then ${P}$ splits completely in $E/F$.
\end{itemize}
\end{proposition}
\begin{proof}
Refer to \cite[Proposition 3.7.10]{sti09} and \cite[Corollary 3.3.8]{sti09} for the proof of (i), (ii), and (iv), respectively. The proof of (iii) is similar to the Kummer extension case. By the same arguments, we can see that $y$ has a unique pole at ${P}_{\infty}$. Moreover, for any place ${Q}\neq {Q}_{\infty}$ and any place ${P}$ in $E$ lying above ${Q}$, the different $d({P}\mid {Q} )=\nu_{{P}}(L'(y))=0$. Thus, by \cite[Corollary 3.5.11]{sti09}, $\{1, y,\dots, y^{p^r-1}\}$ is an integral basis of  $E/F$ at ${Q}$.
\end{proof}

Therefore, the Artin-Schreier extension is another class of function fields that satisfy the four properties (P4). In this case, the structure of the Galois group $G=\Gal(E/F)\cong W$ and the conjugates of $y$ are also clear. We can describe all subgroups of $G_i$, $E_i=E^{G_i}$, and $\cB$ explicitly.

Assume $W_i=(\ZZ/p\ZZ)w_1+(\ZZ/p\ZZ)w_2+\dots+(\ZZ/p\ZZ)w_i$ is an $\F_p$-subspace of $W$ of dimension $i$ for $1\leq i\leq r$. Then
\begin{enumerate}
\item $G_i:=G_{W_i}=\{\sigma_w\ |\ \sigma_w(y)=y+w,\ \text{for\ all}\ w\in W_i\}$, $i=1,\dots, r$.
\item $E_i=F(y_i)$ and \[y_i=\prod_{w\in W_i} (y+w)=\ell_i(y),\] where $\ell_i(T)\in\F_q[T]$ is a linearized polynomial of degree $p^i$ and $\ell_i(T)$ has all roots in $W_i\subset\F_q$. Moreover, by $W_i=W_{i-1}+(\ZZ/p\ZZ)w_i$, we have
 \[\begin{split}
 y_i&=\prod_{w\in W_i} (y+w)=\prod_{w\in W_{i-1}}\prod_{a=0}^{p-1}\sigma_{w+aw_i}(y)=\prod_{w\in W_{i-1}}\sigma_{w}\left(\prod_{a=0}^{p-1}\sigma_{aw_i}(y)\right)\\
 &=\prod_{w\in W_{i-1}}\sigma_{w}( y^{p}-yw_i^{p-1})=\ell_{i-1}(y^p-w_{i}^{p-1}y)=y_{i-1}^p-\ell_{i-1}(w_i^{p-1})\cdot y_{i-1}.\end{split}\]
 Namely, each $y_i$ is a linearized polynomial of $y_{i-1}$ of degree $p$ for all $i$.
\item A basis of $\cL(m{P}_{\infty})$ is
\begin{equation}\label{eq:atbases}
\cB= \left\{y^{i_0}\ell_1(y)^{i_1}\cdots \ell_{r-1}(y)^{i_{r-1}}\cdot x^j \mid i\cdot d+j\cdot p^r\leq m,\,0\leq i\leq p^r-1,\ j\geq 0\right\},
\end{equation}
where $(i_0, i_1,\ldots, i_{r-1})$ is defined by the $p$-adic expansion $\sum_{k=0}^{r-1} i_kp^k$ of $i\in[0, p^{r}-1]$.
\end{enumerate}


\begin{corollary}\label{cor:at}
Let $E/F$ be the Artin-Schreier extension of degree $p^r$ and ${P}_{\infty}$ be the unique place of $E$ lying over the pole of $x$. Suppose $\cQ=\{\Ga\in \F_q\mid x-\Ga \ \text{completely\ splits\ in}\ E/F\}$ and $|\cQ|=s$. Let $\cP$ be the set of all rational places in $E$ lying above the places of $\cQ$.   {If the MPE of any polynomial in $\F_q[x]_{<s}$ at $\cQ$ can be done in $O(s\log s)$ operations}, then encoding of $C(\cP,\lambda{P}_{\infty})$ with $\Gl<N=ms$ runs in $O(p\cdot N\log N)$ operations of $\F_q$.
\end{corollary}

\subsection{Mixed extensions}
The assumption about $G=\Gal(E/F)$ is abelian in the box (P4) in Section~\ref{sec:FMPE} is not necessary. Actually, we only need there is a chain of subgroups of $G$ for $G$ with smooth order. Then, by the other properties in the box (P4), we can also apply the divide-and-conquer method. In this subsection, we present a non-abelian extension $E/\F_q(x)$ constructed from the affine linear group which ensures the other properties in (P4) hold.

Let $F=\F_q(x)$ be the rational function field. We denote by  $\Aut(F/\F_q)$ the automorphism group of $F$ over $\F_q$, i.e.,
\begin{equation}
\Aut(F/\F_q)=\{\sigma: F\rightarrow F \mid \sigma  \mbox{ is an } \F_q\mbox{-automorphism of } F\}.
\end{equation}
It is clear that an automorphism $\sigma\in \Aut(F/\F_q)$ is uniquely determined by $\sigma(x)$.
It is well known that every automorphism $\sigma\in \Aut(F/\F_q)$ is given by
\begin{equation}\label{abcd}
\sigma(x)=\frac{ax+b}{cx+d}
\end{equation}
for some constant $a,b,c,d\in\F_q$ with $ad-bc\neq0$ (see \cite{hir08}).
Denote by  $\GL_2(q)$ the general linear group of $2\times 2$ invertible matrices over $\F_q$.
Thus, every matrix $A=\left(\begin{array}{cc}a&b\\ c&d\end{array}\right)\in \GL_2(q)$ induces an automorphism of $F$ given by \eqref{abcd}.
Two matrices of $\GL_2(q)$ induce the same automorphism of $F$ if and only if they belong to the same coset of $Z(\GL_2(q))$, where $Z(\GL_2(q))$ stands for the center $\{aI_2:\; a\in\F_q^*\}$ of $\GL_2(q)$.  This implies that $\Aut(F/\F_q)$ is isomorphic to the projective linear group $\PGL_2(q):=\GL_2(q)/Z(\GL_2(q))$. Thus, we can identify $\Aut(F/\F_q)$ with $\PGL_2(q)$.

Consider the affine linear subgroup $\AGL_2(q)$ of $\PGL_2(q)$
\begin{equation}\label{eq:x3}
\AGL_2(q):=\left\{\left(\begin{array}{cc}a&b\\ 0&1\end{array}\right):\; a\in\F_q^*,b\in\F_q\right\}.
\end{equation}
Every element $A=\left(\begin{array}{cc}a&b\\ 0&1\end{array}\right)\in \AGL_2(q)$ defines an affine automorphism
$\sigma(x)=ax+b.$  We identify each $A$ with $\sigma$ in this way.

Let $G$ be a subgroup of $\AGL_2(q)$ and let $A(y)=\prod_{\sigma\in G}\sigma(y)$. If $A(y)=u(x)$ is absolutely irreducible and separable over $\F_q(x)$ for some polynomial $u(x)\in\F_q[x]$, then $E=\F_q(x, y)$ is a Galois extension over $\F_q(x)$ with $\Gal(E/F)=G$. Moreover, the conjugate roots of $y$ are exactly $\sigma(y)=ay+b$ for all $\sigma\in G$. Actually, the Kummer extension corresponds to $G\leq \left\{\left(\begin{array}{cc}a&1\\ 0&1\end{array}\right):\; a\in\F_q^*\right\}$ , while the Artin-Schreier extension corresponds to $G\leq \left\{\left(\begin{array}{cc}1&b\\ 0&1\end{array}\right):\; b\in\F_q\right\}$. In the following, we consider the mixed case.

Assume $\F_q=\F_{\kappa^r}$, where $\kappa$ is a power of prime $p$. Let $T\leq \F_{\kappa}^*$ be a multiplicative subgroup of order $e$ and $W\leq \F_{\kappa^r}$ be an additive subgroup of order $\kappa^w$, respectively. Let $G$ be the group of semidirect product of $T$ and $W$,
\[G=\left\{\left(\begin{array}{cc}a&b\\ 0&1\end{array}\right)\mid a\in T, b\in W\right\}=T\ltimes W.\]
Then $G$ is a non-abelian subgroup of $\AGL_2(\F_q)$ of order $m=\kappa^we$ and $\{1\}\ltimes W$ is a normal subgroup of $G$. Assume $e=\prod_{i=1}^t p_i$. Then $T$ has a subgroup chain: $\{1\}=T_0\leq T_1\ldots\leq T_t=T$, where $[T_i: T_{i-1}]=p_i$ for $i=1,\ldots, t$. Moreover, $W$ also has a subgroup chain: $\{0\}=W_0 \trianglelefteq W_1 \trianglelefteq \ldots  \trianglelefteq W_w=W$, where $[W_i: W_{i-1}]=\kappa$ for $i=1,\ldots, w$. Therefore, we can construct a normal subgroup chain of $G$:
\[\{0\}=G_0\trianglelefteq G_1\ldots \trianglelefteq G_w\trianglelefteq G_{w+1}\trianglelefteq \ldots \trianglelefteq G_{t+w}=G,\]
where $G_i=\{1\}\ltimes W_i$ for $0\leq i\leq w$ and $G_{w+j}=T_j\ltimes W$ for $0\leq j\leq t$. Then, by the Galois theory, we can construct a subfield tower of $E/F$:
\[E=E^{G_0}\supsetneq E^{G_1}\supsetneq \ldots\supsetneq E^{G}=F.\]
Moreover, if $\gcd\left(\deg(u(x)), m\right)=1$, then the pole place of $x$ in $\PP_{F}$ is totally ramified. We present the following result without proof as it is similar to the ones given in Subsections 4.1 and 4.2.

\begin{theorem}\label{thm:mix}
Let $\F_q=\F_{\kappa^r}$ be a finite field with $\kappa$ a prime power. Assume $G$ is a semidirect product subgroup of $\AGL_2(q)$ of order $m=\kappa^we$, where $e\mid (\kappa-1)$. Let $A(y)=\prod_{\sigma\in G}\sigma(y)$. Suppose $u(x)\in\F_q[x]$ is a square-free polynomial satisfying:
\begin{itemize}
\item[(i)] $A(y)-u(x)$ is absolutely irreducible and separable over $\F_q(x)$;
\item[(ii)] $\gcd(\deg(A(y),\deg(u(x))=1$. \end{itemize}
Let \[E=\F_q(x,y),\ \text{with}\ A(y)=u(x).\]
Then the pole place $Q_{\infty}\in \PP_{\F_q(x)}$ of $x$ is totally ramified in $E/F_q(x)$. Let $P_{\infty}$ denote the unique place in $\PP_E$ lying above $Q_{\infty}$. Suppose there are $s=\Theta(q)$ rational places in $\PP_{\F_q(x)}\setminus\{Q_{\infty}\}$ splitting completely in $E/\F_q(x)$. Let $\cP\subset \PP_E$ be the set of all rational places lying above these $s$ places. Then for any positive integer $\lambda$ less than $N=m\cdot s$, encoding of $C(\cP,\lambda{P}_{\infty})$ runs in $O(\kappa\cdot N\log N)$ operations of $\F_q$.
\end{theorem}

\section{Examples for codes from plane curves}
In this section, we illustrate our encoding algorithms by using algebraic geometry codes from Hermitian and norm-trace curves. 

\begin{example}{(Hermitian curve)}\label{ex:hc}
Let $\kappa$ be a power of a prime $p$. Let $\F_q=\F_{\kappa^2}$ be a finite field. Consider the Hermitian curve $\H$ over $\F_q$ defined by
\[\H(x,y):=y^{\kappa}+y-x^{\kappa+1}.\]
Then $\H(x,y)$ is absolutely irreducible. Let $E=\F_q(x, y)=\mathrm{Frac}(\F_q[x,y]/\H(x,y))$ be the fraction field. On the one hand, $E$ can be viewed as a Kummer extension over $\F_q(y)$ by extending $\F_q(y)$ with a variable $x$ satisfying the relation $\H(x,y)=0$. On the other hand, $E$ can also be viewed as an Artin-Schreier extension over $\F_q(x)$ by extending $\F_q(x)$ with a variable $y$ satisfying the relation $\H(x,y)=0$. This brings more flexibility of $E$ to be an FMPE-friendly function field.

Consider the one-point Riemann-Roch space $\cL(\Gl{P}_{\infty})$, where $\Gl$ is a positive integer and ${P}_{\infty}$ is the unique pole of both $x$ and $y$. The one-point Hermitian code is the image of the multipoint evaluation of $\cL(\Gl{P}_{\infty})$ at a set $\calP$, denoted by $C(\calP, \Gl{P}_{\infty})$. According to whether $q$ or $q-1$ is smooth, we discuss encoding of $C(\calP, \Gl{P}_{\infty})$ in two different cases.

\textbf{Case 1: $\F_q$ has constant characteristic $p$. } Consider the Artin-Schreier extension $\F_q(x,y)/\F_q(x)$. For any $\alpha\in \F_q$, the equation $y^{\kappa}+y=\alpha^{\kappa+1}$ has $\kappa$ solutions in $\F_q$. Thus there are $q$ rational places of $\F_q(x)$ that are splitting completely in $E/\F_q(x)$. Let $\calP$ denote the set of rational places of $E$ lying above them. Then $\calP=N(E)\setminus\{{P}_{\infty}\}$ and $N=|\calP|=\kappa^3$. Since the MPE of functions in $\F_q[x]_{<q}$ at $\F_q$ can be done in $O(q\log q)$\cite{lin2014novel}, by Corollary~\ref{cor:at}, the Hermitian code $C(\calP, \Gl{P}_{\infty})$ with $\Gl = \kappa(q-1)=\kappa^3-\kappa$ is $O(N\log N)$ encodable.


\textbf{Case 2: $q-1$ is smooth.}  As $(\kappa+1)\mid (q-1)$ is also smooth, we take the Kummer extension $\F_q(x,y)/\F_q(y)$. Let $\Ker(\Tr_{q/\kappa})=\{\alpha\in \F_q \mid \alpha^{\kappa}+\alpha=0\}$ and $\cQ=\F_q\setminus\Ker(\Tr_{q/\kappa})$.  Then $|\cQ|=q-\kappa$. By the multiplicative FFT \cite{Pollard71}, the MPE of any function in $\F_q[x]_{<q-\kappa}$ at $\cQ$ costs $O((q-1)\log (q-1))=O((q-\kappa)\log (q-\kappa))$ operations in $\F_q$. For any $\beta\in\cQ$, the equation $x^{\kappa+1}=\beta^{\kappa}+\beta$ has $\kappa+1$ solutions in $\F_q$. Thus all rational places in $\cQ$ are splitting completely in $E/\F_q(y)$. Let $\calP$ be the set of all rational places lying above $\cQ$. Then $N=|\calP|=\kappa^3-\kappa$. By Corollary~\ref{cor:k}, the Hermitian code $C(\calP, \lambda{P}_{\infty})$ with $\lambda=(\kappa+1)(q-\kappa)=\kappa^3-2\kappa-1$ is $O(N\log N)$ encodable.

\end{example}

\begin{example}\label{ex:5.2}{(Norm-trace  curves)}
Many other examples are also included in our class of curves presented in Section 3. For instance, elliptic curves, hyperelliptic curves, coverings of the Hermitian curves, norm-trace curves,  etc. Let us discuss norm-trace curves in detail. 

Let $q=\kappa^r$ and let $W$ be the kernel of the trace function from $\F_q$ to $\F_\kappa$, i.e., $W$ is the $\F_\kappa$-space $\{\Ga\in\F_q:\; \Ga+\Ga^\kappa+\cdots+\Ga^{\kappa^{r-1}}=0\}$. Let $V$ is an $\F_\kappa$-subspace of $W$ and let $e$ be a divisor of $\frac{\kappa^r-1}{\kappa-1}$. A norm-trace curve is defined by either
\begin{equation}\label{e:x}
\cX:\quad \prod_{\Ga\in V}(y-\Ga)=x^{\kappa^{r-1}+\cdots+\kappa+1},
\end{equation}
or
\begin{equation}\label{e:x}
{\mathcal Y}:\quad \prod_{\Ga\in W}(y-\Ga)=y+y^\kappa+\cdots+y^{\kappa^{r-1}}=x^{e}.
\end{equation}

The curve $\cX$ has genus $\frac12(\kappa^{r-1}+\cdots+\kappa)(|V|-1)$. Furthermore, for every $\Gb\in\F_q$, $\Gb^{\kappa^{r-1}+\cdots+\kappa+1}$ is an element of $\F_\kappa$. Thus, there are $|V|$ rational points of $\cX$ lying over $x-\Gb$. Together with the unique common pole $P_\infty$ of $x$ and $y$, $\cX$ has $q|V|+1$ rational points in total. Let $\cP$ be the set of all rational points of $\cX$ except for $P_\infty$, then the code $C(\cP,\Gl P_\infty)$ has length $N=q|V|$. In this case, we consider the Artin-Schreier extension $\F_q(x,y)/\F_q(x)$. If $q$ has a constant characteristic, then encoding of the algebraic geometry code $C(\cP,\Gl P_\infty)$ with $\lambda<q|V|$ costs $O(N\log N)$ operations.

The curve ${\mathcal Y}$ has genus $\frac12(\kappa^{r-1}-1)(e-1)$. Furthermore, for every $\Gb\in\F_q$, $\Gb+\Gb^\kappa+\cdots+\Gb^{\kappa^{r-1}}$ is an element of $\F_\kappa$. Thus, there are $e$ rational points of ${\mathcal Y}$ lying over $y-\Gb$ for $\Gb\not\in W$ and only one point of ${\mathcal Y}$ lying over $y-\Gb$ for $\Gb\in W$. Together with the unique common pole $P_\infty$ of $x$ and $y$, ${\mathcal Y}$ has $(q-\kappa^{r-1})e+1+\kappa^{r-1}$ rational points in total. Let $\cP$ be the set of all rational points of ${\mathcal Y}$ lying over $y-\Gb$ for $\Gb\not\in W$, then the code $C(\cP,\Gl P_\infty)$ has length $N=(q-\kappa^{r-1})e$. In this case, we consider the Kummer extension $\F_q(x,y)/\F_q(y)$. If $q-1$ is smooth, then encoding of the algebraic geometry code $C(\cP,\Gl P_\infty)$ with $\lambda<e(q-\kappa^{r-1})$ runs in $O(N\log N)$ operations.
\end{example}

\section{Example for AG codes from a non-plane curve}

So far, all examples discussed in Section 5 are algebraic geometry codes based on plane curves. As our general encoding algorithm given in Section 3 works for non-plane curves as well, we present an example of encoding of algebraic geometry codes based on non-plane curves in this section. 

The non-plane curve that we are discussing in this section is the Hermitian tower given in below. The main technique for this example is to construct a tower of function fields via the Galois theory for each FMPE-friendly function field. 

Let $\Gk$ be a prime power and let $q=\Gk^2$.  The Hermitian tower of function fields was first introduced and discussed in \cite{Shen93}. Let $F_1=\F_q(x_1)$.
Then the Hermitian tower is defined by the following recursive equations
\begin{equation}\label{eq:x2}x_{i+1}^{\Gk}+x_{i+1}=x_i^{{\Gk}+1},\quad i=1,2,\dots,n-1.\end{equation}
Put  $F_i=\F_q(x_1,x_2,\dots,x_{i})$ for $i\ge 2$. We fix an integer $n$ satisfying $2\le n\le \Gk/2$.

Let us discuss the number of rational places of the function field $F_n$. 
The pole $\Pin$ of $x_1$ in $F_1$ is totally ramified in the extension $F_n/F_1$. Let $P_\infty^{(n)}$ be the unique place of $F_n$ lying over $\Pin$. 
The other $\Gk^{n+1}$ rational places come from the rational places lying over the unique zero $P_\Ga$ of $x_1-\Ga$ for each $\Ga\in\F_q$. Note that for every $\Ga\in\F_q$, $P_\Ga$ splits completely in $F_n$, i.e., there are $\Gk^{n-1}$ rational places lying over $P_\Ga$.
Intuitively, one can think of the rational places of $F_n$ (besides $P_\infty^{(n)}$) as being given by $n$-tuples $(\Ga_1,\Ga_2,\dots,\Ga_n)\in \F_q^n$ that satisfy $\Ga_{i+1}^{\Gk}+\Ga_{i+1}=\Ga_i^{{\Gk}+1}$ for $i=1,2,\dots,n-1$. For each value of $\Ga \in \F_q$, there are precisely $\Gk$ solutions to $\beta \in \F_q$ satisfying $\beta^{\Gk} + \beta = \Ga^{{\Gk}+1}$, so the number of such $n$-tuples is ${\Gk}^{n+1}$ ($q={\Gk}^2$ choices for $\Ga_1$, and then ${\Gk}$ choices for each successive $\Ga_i$, $2 \le i  \le n$). 

The genus $g_n$ of the function field $F_n$  is given by
\begin{equation}
\label{eq:x4}
g_n=\frac12\left(\sum_{i=1}^{n-1}{\Gk}^n\left(1+\frac1{\Gk}\right)^{i-1}-({\Gk}+1)^{n-1}+1\right)\le \frac {{\Gk}^n}{2}\sum_{i=1}^n{n\choose i}\frac1{{\Gk}^{i-1}} \le \frac{n {\Gk}^n}{2} \sum_{i=1}^n \left(\frac{n}{{\Gk}}\right)^{i-1} \le n {\Gk}^n
\end{equation}
where the second and the last inequalities used ${n\choose i}\le n^i$ and ${\Gk} \ge 2n$, respectively, while the first inequality used the following
\[g_n\le \frac{{\Gk}^n}2\sum_{i=1}^{n-1}\left(1+\frac1{\Gk}\right)^{i-1}=\frac{{\Gk}^n}2\times\frac{\left(1+\frac1{\Gk}\right)^{n-1}-1}{\left(1+\frac1{\Gk}\right)-1}=\frac{{\Gk}^{n+1}}2 \sum_{i=1}^{n-1}{n-1\choose i}\frac{1}{{\Gk}^i}\le \frac{{\Gk}^{n}}2 \sum_{i=1}^{n-1}{n\choose i}\frac{1}{{\Gk}^{i-1}}.\]

For an integer $\Gl>0$, the Riemann-Roch space $\cL(\Gl P_\infty^{(n)})$  of $F_n$ has a nice structure that fits well for our encoding algorithm. More precisely speaking,
a basis of $\cL(\Gl P_\infty^{(n)})$ over $\F_q$ can be explicitly constructed as follows
\begin{equation}\label{eq:x3}\left\{x_1^{j_1}\cdots x_n^{j_n}:\; (j_1,\dots,j_n)\in\mathbb{N}^n,\ \sum_{i=1}^nj_i\Gk^{n-i}({\Gk}+1)^{i-1}\le \Gl, 0\le j_2,\cdots,j_n\le \kappa-1\right\}
\end{equation}
from \cite[Proposition 5]{Shen93}.

The above recursive equations are defined in terms of Artin-Schreier extensions.  Let us now define the same tower in terms of Kummer extensions. 
\begin{equation}\label{eq:x7}y_{i+1}^{{\Gk}+1}= y_{i}^{\Gk}+y_{i},\quad i=1,2,\dots,n-1.\end{equation}
Put $E_i=\F_q(y_1,y_2,\dots,y_i)$. Apparently, $E_i$ and $\F_i$ are $\F_q$-isomorphic (in fact, $E_i$ can be obtained from $F_i$ by substitution of variables $x_1,x_2,\dots,x_i$ with $y_1,y_2,\dots,y_i$). Thus, they have the same genus and number of rational places.

As studied for the Hermitian curves,  we discuss encoding of $C(\calP^{(n)}, \Gl_n{P}_{\infty}^{(n)})$ in two different cases.

\textbf{Case 1: $\F_q$ has constant characteristic $p$. } Let us consider the Hermitian tower $\{F_i\}$ of function fields given in \eqref{eq:x2} and keep the same notations.  Let $P_{\infty}^{(i)}$ be the unique place of $F_i$ lying over the pole of $x_1$.
Now let us focus on the Riemann-Roch space $\cL(\lambda_nP_{\infty}^{(n)})$ for an integer $n\ge 2$ and $\lambda_n\in \mathbb{N}$.
Let $\mathcal{P}^{(i)}$ denote the set of all rational places of $F_i$ except for the place lying over the pole of $x_1$. Let $N_i$ denote the size of $\cP^{(i)}$. Then $N_i=\Gk^{i+1}$ for all $1\le i\le n$.

Consider the Artin-Schreier extension $F_i(x_{i+1})/F_i$. It is a Galois extension with the Galois group $\Gal(F_i(x_{i+1})/F_i)\simeq\F_p^r$, where $\Gk=p^r$. Hence, we get a Galois tower with each extension degree equal to $p$ according to Section~\ref{sec:FMPE}. Then all four properties in the box (P4) are satisfied. First of all, it is easy to see the properties (i)-(iii) in (P4) are satisfied. To verify that the last property of (P4) is also satisfied, we can just apply \cite[Theorem 3.5.10]{sti09}, which is the same as in the proof of Proposition~\ref{prop:at}(iii). Thus, encoding of the algebraic geometry code $C(\calP^{(n)}, \Gl_n{P}_{\infty}^{(n)})$ can be reduced encoding of the algebraic geometry code $C(\calP^{(n-1)}, \lfloor\Gl_n/\Gk\rfloor {P}_{\infty}^{(n-1)})$ under a basis constructed from Lemma~\ref{lem:bases}. By recursive reduction, the encoding of $C(\calP^{(n)}, \Gl_n{P}_{\infty}^{(n)})$ is eventually reduced to the encoding of an algebraic geometry code defined over the rational function field $F_1=\F_q(x_1)$ which 
is a Reed-Solomon code. Let $\fC(\lambda_n, N_n)$ denote the number of operations in $\F_q$ required for encoding of a message. Then we have the following recursive formula:

\[\begin{split}
\fC(\lambda_n, N_n)&=\Gk\cdot \fC(\lfloor \lambda_n/{\Gk}\rfloor, N_{n-1})+(p+\cdots +p)O(N_n)\\
&=\Gk \cdot \fC(\lfloor \lambda_n/{\Gk}\rfloor, N_{n-1}) +O(p\cdot N_n\cdot\log \Gk).
\end{split}\]

Inductively to the last step, we obtain 
\begin{eqnarray*}
\fC(\lambda_n, N_n)&=&\Gk \cdot \fC(\lfloor \lambda_n/{\Gk}\rfloor, N_{n-1}) +O(p\cdot N_n\cdot\log \Gk)\\
&=&\Gk\left(\Gk \cdot \fC(\lfloor \lambda_n/{\Gk^2}\rfloor, N_{n-2}) +O(p\cdot N_{n-1}\cdot\log \Gk)\right)+O(p\cdot N_n\cdot\log \Gk)\\
&\cdots&\cdots\\
&=&\Gk^{n-1}\fC(\lfloor \lambda_n/{\Gk^{n-1}}\rfloor, N_{1})+O(p\log\Gk(\Gk^{n-2}N_2+\cdots+\Gk N_{n-1}+N_n))\\
&=&\Gk^{n-1}O(p\cdot N_1\log N_1)+O(pN_n\log\Gk^n )\\
&=&O(N_n\log N_n ),
\end{eqnarray*}
where we use the two facts: (i) $\fC(\lfloor \lambda_n/{\Gk^{n-1}}\rfloor, N_{1})$ is in fact the encoding complexity of a Reed-Solomon code and hence $\fC(\lfloor \lambda_n/{\Gk^{n-1}}\rfloor, N_{1})=O(p\cdot N_1\log N_1)$ \cite{lin2014novel}; (ii) $N_i=\Gk^{i+1}=N_n/{\Gk^{n-i}}$. The characteristic $p$ disappears in the last equality due to the fact that $p$ is constant.  

The last thing worth mentioning is that the required basis of $\cL\left(\Gl_n {P}_{\infty}^{(n)}\right)$ to perform the FMPE can be constructed recursively by Lemma~\ref{lem:bases}, namely, it consists of multivariate polynomials $\bx_{1}^{\bj_1}\bx_{2}^{\bj_{2}}\cdots \bx_n^{\bj_n}$, where $j_1\geq 0, 0\leq j_2,\ldots, j_n<\kappa$, the bold $\bj_1,\ldots,\bj_n$ stand for the vectors from the $p$-adic expansions of $j_1,\ldots,j_n$, respectively, and $\bx_i$ is the product of linearized polynomial of $x_i$ defined as in Equation \eqref{eq:atbases}. Then every message $\bm\in \F_q^{k_n}$ with $k_n=\dim\left(\cL(\Gl_n {P}_{\infty}^{(n)})\right)$ is mapped to a function $f_{\bm}\in\cL\left(\Gl_n {P}_{\infty}^{(n)}\right)$. Hence $\bm$ is encoded as a codeword $\ev_{\calP}(f_{\bm})$ using $O(N_n\log N_n)$ operations of $\F_q$.

\textbf{Case 2: $q-1$ is smooth.}  In this case, we consider the Kummer extension $E_{i+1}/E_i$. By the same discussion in Example~\ref{ex:hc}(Case 2), we can show that every extension $E_{i+1}/E_i$ satisfies all four properties in the box (P4). Thus, the MPE of functions in $\cL(\lambda_nP_{\infty}^{(n)})$ can be recursively reduced to the MPE of functions in $\cL(\lfloor \lambda_n/{\Gk^{n-1}}\rfloor P_{\infty}^{(1)})$. 
In this case, the basis of $\cL(\lambda_nP_{\infty}^{(n)})$ of $E_n$ can be given explicitly as
\[\left\{y_1^{j_1}y_2^{j_2}\cdots y_n^{j_n}:\; (j_1,\dots,j_n)\in\mathbb{N}^n,\ \sum_{i=1}^nj_i\Gk^{i-1}({\Gk}+1)^{n-i}\le \Gl, 0\le j_2,\cdots,j_n\le \kappa\right\}.\]
Finally, in the same way, we can show that algebraic geometry codes from the function field $E_n$ also have encoding complexity $O(N \log N )$.

\bibliographystyle{alpha}
\bibliography{sjtu_li}

\end{document}